\pdfoutput = 1

\documentclass[english]{article}
\usepackage{lmodern}
\usepackage{lmodern}
\usepackage[T1]{fontenc}
\usepackage[latin9]{luainputenc}
\usepackage{geometry}
\usepackage{color}
\usepackage{babel}
\usepackage{array}
\usepackage{float}
\usepackage{multirow}
\usepackage{amsmath}
\usepackage{amsthm}
\usepackage{amssymb}
\usepackage{graphicx}

\usepackage {hyperref}
\hypersetup{
unicode=true,pdfusetitle,
bookmarks=true,bookmarksnumbered=false,bookmarksopen=false,
breaklinks=false,pdfborder={0 0 1},backref=false,colorlinks=false
}

\makeatletter

\providecommand{\tabularnewline}{\\}
\floatstyle{ruled}
\newfloat{algorithm}{tbp}{loa}
\providecommand{\algorithmname}{Algorithm}
\floatname{algorithm}{\protect\algorithmname}

\theoremstyle{plain}
\newtheorem{thm}{\protect\theoremname}[section]
  \theoremstyle{definition}
  \newtheorem{defn}[thm]{\protect\definitionname}
  \theoremstyle{definition}
  \newtheorem{example}[thm]{\protect\examplename}
  \theoremstyle{remark}
  \newtheorem{rem}[thm]{\protect\remarkname}
  \theoremstyle{plain}
  \newtheorem{prop}[thm]{\protect\propositionname}
  \theoremstyle{plain}
  \newtheorem{assumption}[thm]{\protect\assumptionname}
  \theoremstyle{remark}
  \newtheorem*{acknowledgement*}{\protect\acknowledgementname}

\usepackage{babel}
\usepackage{pgfplots}
\usepackage{amsfonts}
\usepackage{capt-of}
\usepackage{tikz}
\usepackage{tikz-3dplot}\pgfplotsset{compat=1.6}

\pgfplotsset{soldot/.style={only marks,mark=*}}
\pgfplotsset{holdot/.style={fill=white,only marks,mark=*}}
\pgfplotsset{tick label style={color=white},
label style={font=\small},
legend style={font=\small}}

  \providecommand{\acknowledgementname}{Acknowledgement}
  \providecommand{\assumptionname}{Assumption}
  \providecommand{\definitionname}{Definition}
  \providecommand{\examplename}{Example}
  \providecommand{\propositionname}{Proposition}
  \providecommand{\remarkname}{Remark}
\providecommand{\theoremname}{Theorem}

\makeatother

  \providecommand{\acknowledgementname}{Acknowledgement}
  \providecommand{\assumptionname}{Assumption}
  \providecommand{\definitionname}{Definition}
  \providecommand{\examplename}{Example}
  \providecommand{\propositionname}{Proposition}
  \providecommand{\remarkname}{Remark}
\providecommand{\theoremname}{Theorem}

\begin{document}
\title{Preference elicitation and robust optimization with multi-attribute quasi-concave
choice functions}

\author{William B. Haskell \footnote{William B. Haskell (isehwb@nus.edu.sg) is an assistant professor in Department of Industrial Systems Engineering and Management at National University of Sinagpore},\quad Wenjie Huang \footnote{Wenjie Huang (wenjie\_huang@u.nus.edu) is a Ph.D. Candidate in Department of Industrial Systems Engineering and Management at National University of Sinagpore.} \quad and Huifu Xu \footnote{Huifu Xu (h.xu@soton.ac.uk) is a professor in School of Mathematical Sciences at University of Southampton.}}

\maketitle

\begin{abstract}
Decision maker's preferences are often captured by some choice functions
which are used to rank prospects. In this paper, we consider ambiguity
in choice functions over a multi-attribute prospect space. Our main
result is a robust preference model where the optimal decision is
based on the worst-case choice function from an ambiguity set constructed
through preference elicitation with pairwise comparisons of prospects.
Differing from existing works in the area, our focus is on quasi-concave
choice functions rather than concave functions and this enables us
to cover a wide range of utility/risk preference problems including
multi-attribute expected utility and $S$-shaped aspirational risk
preferences. The robust choice function is increasing and quasi-concave
but not necessarily translation invariant, a key property of monetary
risk measures. We propose two approaches based respectively on the
support functions and level functions of quasi-concave functions to
develop tractable formulations of the maximin preference robust optimization
model. The former gives rise to a mixed integer linear programming
problem whereas the latter is equivalent to solving a sequence of
convex risk minimization problems. To assess the effectiveness of
the proposed robust preference optimization model and numerical schemes,
we apply them to a security budget allocation problem and report some
preliminary results from experiments. 
\end{abstract}

\section{Introduction }

Decision making under uncertainty is a universal theme in the stochastic
and robust optimization communities. Much of the focus has been on
exogenous uncertainties which go beyond the control of the decision
maker such as market demand and climate change. In practice, however,
there is often significant endogenous uncertainty which arises from
ambiguity about the decision maker's preferences; either because the
decision making involves several stakeholders who are unable to reach
a consensus, or there is inadequate information for the decision maker
to identify a unique utility/risk function which precisely characterizes
his preferences. Preference robust optimization (PRO) models are subsequently
proposed where the optimal decision is based on the worst preference
from an ambiguity set of utility/risk preference functions constructed
through available partial information. The robustness is aimed to
mitigate the risk arising from ambiguity about the decision maker's
preferences. 

Armbruster and Delage \cite{armbruster2015decision} study a PRO model
for utility maximization problems. Specifically, they model decision
maker's ambiguity about utility preferences by incorporating various
possible properties of the utility function such as monotonicity,
concavity, and S-shapedness, along with preference elicitation information
from pairwise comparisons. An important component their research is
to derive tractable formulations of the resulting maximin problem
and they manage to do so by exploiting linear envelopes of convex/concave
functions. \textcolor{black}{Delage and Li} \cite{delage2017minimizing}
extend this research to risk management problem in finance where the
investor's choice of a risk measure is ambiguous. As in \cite{armbruster2015decision},
they consider the ambiguity set of risk measures primarily via pairwise
elicitation but also featured with important properties such as convexity,
coherence, and law invariance, and they develop tractable formulations
accordingly. 

Hu and Mehrotra \cite{hu2015robust} approach the PRO model in a different
manner. First, they propose a moment-type approach which allows one
to define the ambiguity set for a decision maker's utility preferences
via the certainty equivalent method, pairwise comparisons, upper and
lower bounds of the trajectories of the utility functions, and bounds
on their derivatives at specified grid points; second, they consider
a probabilistic representation of the class of increasing convex utility
functions by confining them to a compact interval and scaling them
to being bounded by one; third, by constructing a piecewise linear
approximation of the trajectories of the utility bounds, they derive
a tractable reformulation of the resulting PRO as a linear programming
problem. Qualitative convergence analysis is presented to justify
the piecewise linear approximation.

Hu and Mehrotra's approach is closely related to stochastic dominance,
a subject which has been intensely studied over the past few decades,
see the monographs \cite{Muller2002,Shaked2007} for a comprehensive
treatment of the topic and \cite{Dentcheva2003,Dentcheva2004} for
the optimization models with stochastic dominance constraints. Indeed,
when the preference of a decision maker satisfies certain axioms including
completeness, transitivity, continuity and independence, Von Neumann
and Morgenstern's expected utility theory\textcolor{black}{{} (\cite{von1945theory})}
guarantees that any set of preferences that the decision maker may
have among uncertain/risky prospects can be characterized by an expected
utility measure. The issue that we are looking at here is that the
decision maker does not necessarily have complete information about
his preferences, and prospects associated with decisions do not necessarily
have definitive stochastic dominance relationships - which means that
the existing models based on stochastic dominance are not applicable.

In a more recent development, Haskell et al. \cite{Haskell_Robust}
consider a stochastic optimization problem where the decision maker
faces both exogenous uncertainty and endogenous uncertainty associated
with decision maker's risk attitude. They propose a PRO model where
the ambiguity is constructed in the product space of exogenous uncertainty
and utility functions. By using Lagrangian duality, they derive an
exact tractable reformulation for some special cases and a tractable
relaxation in the general setting. Delage et al. \cite{delage2017shortfall}
propose a robust shortfall risk measure model to tackle the case where
investors are ambiguous about their utility loss functions. They study
viable ways to identify the ambiguity set of loss functions via pairwise
comparison for utility risk measures with respective features such
as coherence, convexity, and boundedness, and they derive tractable
linear program reformulation for the resulting optimization problems
with robust shortfall risk constraints.

In this paper, we take on this stream of research but with a different
focus in terms of both modeling and tractable reformulations. Specifically,
we consider a class of so-called choice functions applicable to multi-attribute
decision making which are monotonically increasing along some specified
direction, quasi-concave, but not necessarily convex/concave or translation
invariant. We tackle the subsequent PRO model via ``hockey-stick''
type support functions and level functions. Monotonicity ensures that
the decision maker universally prefers more to less. Quasi-concavity
is a general form of risk aversion, the axiom of quasi-concavity is
further supported by the study in \cite{brown2012aspirational} on
aspirational preferences. Moreover, by dropping translation invariance)
which is an key axiomatic property of convex risk measures), we allow
our PRO model to cover a wider range of problems where translation
invariance may not hold \cite{brown2012aspirational}.

Another important departure from existing research is that our PRO
model is applicable to multi-attribute decision making problems. These
problems are ubiquitous in practical applications but the existing
PRO models mainly emphasize the single attribute setting. For instance,
in management research of healthcare, it is typical to use several
metrics rather than just one to measure the quality of life \cite{torrance1982application,feeny2002multiattribute,smith2005your}.
Similar problems can be found in network management problems \cite{azaron2008multi,chen2010stochastic},
scheduling, \cite{liefooghe2007combinatorial,Zakariazadeh2014}, design
\textcolor{black}{\cite{tseng1990minimax,Dino2017}} and portfolio
optimization \cite{fliege2014robust}. Indeed, over the past few decade,
there have been significant research on multi-attribute expected utility
\cite{von1988decomposition,fishburn1992multiattribute,miyamoto1996multiattribute,tsetlin2006equivalent,tsetlin2007decision,Tsetlin_Multiattribute_2009}
and multi-attribute risk management \cite{jouini2004vector,burgert2006consistent,hamel2010duality,Galichon2012}.
In a more recent development, research on robust multi-attribute choice
models has also emerged, see \cite{lam2013multiple,ehrgott2014minmax,noyan2016optimization}
for instance. In particular, \cite{noyan2016optimization} considers
optimization with a general class of scalarization functions (in particular,
the class of min-biaffine functions), where the weights of the scalarization
lie in a convex ambiguity set.

We summarize the main contributions of our present paper as follows.
\begin{itemize}
\item \textbf{Multi-attribute quasi-concave PRO model.} We propose a robust
choice model for preference ambiguity where the underlying choice
function is monotonic and quasi-concave. By replacing concavity with
quasi-concavity and dropping translation invariance, we extend the
existing PRO model so that it is easier to incorporate preference
elicitation information. Moreover, the new model framework covers
a number of well-known preference models (such as expected utility
and aspirational preferences) and can be applied to a wider range
of PRO problems where the decision maker's choice function is merely
increasing and quasi-concave. Of course, it also poses new challenges
to tractable reformulation which so far have depended on convexity
\cite{armbruster2015decision,Haskell_Aspects_2015,delage2017minimizing}.
Our model's support for multi-objective stochastic decision making
problems makes the model applicable to an even broader class of problems.
\item \textbf{New forms of tractable formulations.} We propose two schemes
for tractable reformulation of the proposed new PRO model. One is
based on the support functions of quasi-concave functions and the
other exploits approximation of quasi-concave function by a sequence
of convex level functions. While both approaches are well known in
the literature of generalized convex optimization, they are applied
to PRO here for the first time. The support function approach takes
the PRO model to a mixed-integer linear program as opposed to a linear
program as in earlier work. The level set representation is a generalization
of the representation results in \cite{brown2009satisficing,brown2012aspirational}
to the case of preference ambiguity. We are able to explicitly derive
the connection between these two approaches by using the special form
of piecewise linear support functions for quasi-concave functions.
This framework, based on representing any monotonic diversification
favoring choice function in terms of a family of risk functions, leads
to a unifying framework for representing multi-attribute choice functions.
Moreover, this framework naturally converts a decision-making problem
with a quasi-concave choice function into a sequence of convex optimization
problems, yielding a viable computational recipe. This development
is related to the methods in \cite{brown2009satisficing,brown2012aspirational}
and extends these methods to the multi-attribute setting. The level
set representation further builds on the managerial insights from
\cite{brown2009satisficing,brown2012aspirational}. It reveals that
multi-attribute choice functions can in general be understood in terms
of a family of multi-attribute risk functions and the decision maker's
desired satiation levels.
\item \textbf{Level function method.} In the case where the attributes depend
nonlinearly on the decision variables, we propose an algorithm for
solving the PRO by using the level function method from \cite{Xu2001}.
The algorithm is an iterative regime where at each iterate we solve
a mixed-integer linear program based on the support function representation
and identification of a level function, which is closely related to
the level set representation. To examine the performance of the model
and numerical schemes, we apply them to a homeland security problem
considered by Hu and Mehrotra \cite{Hu2011}. Our numerical experiments
for this problem illustrate how our PRO model captures diversification
favoring behavior, and also how it depends on the elicited comparison
data set.
\end{itemize}
The rest of the paper is organized as follows. In Section 2, we review
some preliminary materials related to choice functions. In Section
3, we formally describe the robust choice function model including
the definition of the class of choice functions to be considered,
specification of the ambiguity set, and characterization of the robust
choice function and the corresponding maximin optimization problem.
Section 4 details the support function approach for tractable reformulation
of the robust choice model and its underlying theory. Next in Section
5, we discuss an alternative level set representation for quasi-concave
choice functions. This development leads to a general representation
formula for quasi-concave choice functions. Here, we also discuss
the role of ``targets'' which feature prominently in the decision
analysis literature. Section 6 builds on our PRO model and explains
how to solve optimization problems in the presence of preference ambiguity,
and Section 7 applies this methodology to a budget allocation problem
for homeland security. The paper concludes in Section 8 with a discussion
of potential impacts and future research directions.

\section{Preliminaries }

This section presents some preliminary materials on choice functions.
We begin with a set of states of nature $\Omega$ endowed with a $\sigma-$algebra
$\mathcal{B}$. Let $\mathcal{L}$ be an admissible space of measurable
mappings $X\mbox{ : }\Omega\rightarrow\mathbb{R}^{n}$, equipped with
the supremum norm topology. We generally treat $\mathcal{L}$ as a
space of multi-attribute prospects with $n\geq2$ characteristics,
although the case $n=1$ is also covered. The inequality $X\leq Y$
for $X,\,Y\in\mathcal{L}$ is understood to mean $X\left(\omega\right)\leq Y\left(\omega\right)$
component-wise for all $\omega\in\Omega$. We adopt the convention
that prospects in $\mathcal{L}$ represent rewards/gains so that larger
values of all attributes are preferred to smaller values. 

Let $\bar{\mathbb{R}}\triangleq\mathbb{R}\cup\left\{ -\infty,\,\infty\right\} $
be the extended-valued real line. A choice function is a mapping $\rho\mbox{ : }\mathcal{L}\rightarrow\bar{\mathbb{R}}$
that gives numerical values to prospects in $\mathcal{L}$ to evaluate
their fitness (e.g. see \cite{brown2009satisficing}). When $\rho\left(X\right)\geq\rho\left(Y\right)$,
$X$ is said to be weakly preferred to $Y$ and when $\rho\left(X\right)>\rho\left(Y\right)$,
$X$ is said to be strongly preferred to $Y$. The following definition
specifies the key properties of choice functions that appear frequently
in the literature (see \cite{Ruszczynski2006a,brown2009satisficing}
for example).
\begin{defn}[Properties of the choice function]
\label{def:choice} 
(i) (Upper semi-continuity) For all $X\in\mathcal{L}$,
$\limsup_{Y\rightarrow X}\rho\left(Y\right)=\rho\left(X\right)$.

(ii) (Monotonicity) For all $X,\,Y\in\mathcal{L}$, $X\leq Y$ implies
$\rho\left(X\right)\leq\rho\left(Y\right)$.

(iii) (Quasi-concavity) For all $X,\,Y\in\mathcal{L}$, $\rho\left(\lambda\,X+\left(1-\lambda\right)Y\right)\geq\min\left\{ \rho\left(X\right),\,\rho\left(Y\right)\right\} $
for all $\lambda\in\left[0,\,1\right]$.

(iv) (Completeness) For all $X,\,Y\in\mathcal{L}$, either $\rho(X)\geq\rho(Y)$
or $\rho(X)\leq\rho(Y)$.
\end{defn}

Upper semi-continuity is a common technical condition (see \cite{Balder1993,brown2012aspirational,Voorneveld2016}).
Monotonicity means that the decision maker always prefers more reward
to less - this property is universally accepted. Quasi-concavity means
that diversification does not decrease reward, where the convex combination
$\lambda\,X+\left(1-\lambda\right)Y$ is understood as a mixture of
the prospects $X$ and $Y$, see \cite{brown2012aspirational,brown2009satisficing}.
Properties (i) - (iii) appear in much of the decision theory literature
(e.g. \cite{brown2012aspirational}). Since the choice function is
real-valued, it is automatically \textit{transitive}: for all $X,\,Y,\,Z\in\mathcal{L},$
if $\rho(X)\geq\rho(Y)$ and $\rho(Y)\geq\rho(Z)$, then $\rho(X)\geq\rho(Z)$.
Property (iv) and transitivity ensure that the choice functions we
consider are ``rational''.

We now give some examples of multi-attribute choice functions to motivate
our discussion and to point out related work on multi-attribute prospects.
Here and throughout, the Euclidean inner product is denoted by $\langle x,\,y\rangle=\sum_{i=1}^{n}x_{i}y_{i}$. 
\begin{example}
\label{exa:Preliminaries} Let $u_{i}$ for $i=1,\ldots,\,n$ be univariate
utility functions.

(i) The assumption of mutual utility independence (see \cite{abbas2015multiattribute}
for example) gives rise to additive utility functions $u\left(x_{1},\ldots,\,x_{n}\right)=\sum_{i=1}^{n}\kappa_{i}u_{i}\left(x_{i}\right)$
for $\kappa_{1},\ldots,\,\kappa_{n}\geq0$. The expected utility of
a random vector $X=\left(X_{1},\ldots,\,X_{n}\right)$ is then $\rho\left(X\right)=\sum_{i=1}^{n}\kappa_{i}\mathbb{E}\left[u_{i}\left(X_{i}\right)\right]$.

(ii) In \cite{keeney1974multiplicative}, an alternative independent
utility aggregation model is proposed with
\[
u\left(x_{1},\ldots,\,x_{n}\right)=\frac{1}{K^{\prime}}\left\{ \left[\prod_{i=1}^{n}\left(K^{\prime}\kappa_{i}u_{i}\left(x_{i}\right)+1\right)\right]-1\right\} ,
\]
where $\sum_{i=1}^{n}\kappa_{i}\neq1,\,K^{\prime}>-1$ and $K^{\prime}+1=\prod_{i=1}^{n}\left(\kappa_{i}K^{\prime}+1\right).$
The expected utility is then 
\[
\rho\left(X\right)=\frac{1}{K^{\prime}}\left\{ \left[\prod_{i=1}^{n}\left(K^{\prime}\kappa_{i}\mathbb{E}\left[u_{i}\left(X_{i}\right)\right]+1\right)\right]-1\right\} .
\]

(iii) For a general utility function $u\mbox{ : }\mathbb{R}^{n}\rightarrow\mathbb{R}$,
the expected utility $\rho\left(X\right)=\mathbb{E}\left[u\left(X\right)\right]$
is a choice function.

(iv) Let $C\mbox{ : }\left[0,\,1\right]^{n}\rightarrow\left[0,\,1\right]$
be a multivariate copula (i.e. $C$ is the joint cumulative distribution
function of an $n-$dimensional random vector on the unit cube $\left[0,\,1\right]^{n}$
with uniformly distributed marginals), then $\rho\left(X\right)=C\left(u_{1}\left(X_{1}\right),\ldots,\,u_{n}\left(X_{n}\right)\right)$
is a choice function (see \cite{abbas2009multiattribute}).

(v) The conditional value-at-risk (CVaR) of a univariate random variable
$X$ at level $\alpha\in\left(0,\,1\right)$ is 
\[
\mbox{CVaR}_{\alpha}\left(X\right)\triangleq\inf_{\eta\in\mathbb{R}}\left\{ \eta+\left(1-\alpha\right)^{-1}\mathbb{E}\left[\left(X-\eta\right)_{+}\right]\right\} .
\]
A multivariate version of the conditional value-at-risk (CVaR) is
developed in \cite{noyan2013optimization} based on linear scalarization.
Given a vector of weights $w\in\mathbb{R}_{+}^{n}$, we may consider
the choice function $\rho\left(X\right)=\mbox{CVaR}_{\alpha}\left(\langle w,\,X\rangle\right)$
on $\mathcal{L}$.

(vi) More generally, as in \cite{noyan2016optimization}, we may take
any univariate risk measure $\vartheta$ (such as a mean-deviation
risk measure, see \cite{Ruszczynski2006a}) and then consider the
choice function $\rho\left(X\right)=\vartheta\left(\varphi\left(w,\,X\right)\right)$
where $\varphi\text{ : }\mathbb{R}^{n}\times\mathbb{R}^{n}\rightarrow\mathbb{R}$
is a scalarization function. In \cite{noyan2016optimization}, the
authors focus on the computationally tractable class of ``min-biaffine''
scalarization functions $\varphi$. 
\end{example}

\section{Robust preferences model}

Now we come to the main problem under consideration in this paper.
To begin, we introduce the set of all upper semi-continuous, monotonic,
and quasi-concave choice functions 
\[
\mathcal{R}_{iqv}\triangleq\left\{ \text{upper semi-continuous, increasing, and quasi-concave }\rho\text{ : }\mathcal{L}\rightarrow\mathbb{R}\right\} .
\]
This set characterizes our decision makers of interest. Since concave
functions are quasi-concave, this class of functions consists of all
continuous increasing concave utility functions in the literature.

In practice it is difficult to elicit a precise functional form for
$\rho$. This difficulty is exacerbated in the multi-attribute setting.
First, when multi-attribute prospects are in play, it is not obvious
how to characterize the marginal dependencies of the variety of attributes,
i.e., it is not always clear how much an increase in the value of
one asset should depend on the levels of the other assets. Second,
it is hard to specify a choice function in group decision making where
the group must come to a consensus. Third, we may have only a few
observations of the decision maker's behavior which makes it impossible
to precisely specify preferences.

To circumvent these difficulties, we design a robust choice function.
To this end, we will first need to construct a preference ambiguity
set $\mathcal{R}\subset\mathcal{R}_{iqv}$ which contains a range
of possible choice functions. Then, given this preference ambiguity
set $\mathcal{R}$, we will set up a framework which chooses a ``robust
choice function'' as:
\begin{equation}
\psi\left(X;\,\mathcal{R}\right)\triangleq\inf_{\rho\in\mathcal{R}}\rho\left(X\right),\,\forall X\in\mathcal{L}.\label{eq:Robust}
\end{equation}
In the upcoming definition, we generalize formulation (\ref{eq:Robust})
by taking a ``benchmark'' prospect $Y\in\mathcal{L}$. Benchmark
prospects have a long history in the field of optimization with stochastic
constraints, see for instance \cite{Dentcheva2004,noyan2013optimization,armbruster2015decision,noyan2016optimization}.
\begin{defn}[Robust choice function]
\label{def:Elicitation} Let $\mathcal{R}\subset\mathcal{R}_{iqv}$
and $Y\in\mathcal{L}$ be given, then $\psi\left(\cdot;\,\mathcal{R},\,Y\right)\text{ : }\mathcal{L}\rightarrow\mathbb{R}$
defined via 
\begin{equation}
\psi\left(X;\,\mathcal{R},\,Y\right)\triangleq\inf_{\rho\in\mathcal{R}}\left\{ \rho\left(X\right)-\rho\left(Y\right)\right\} ,\,\forall X\in\mathcal{L},\label{eq:Robust-1}
\end{equation}
is the robust choice function corresponding to $\mathcal{R}$ and
$Y$. 
\end{defn}

When the benchmark $Y$ is a constant act and all $\rho\in\mathcal{R}$
are normalized to have the same value at $Y$, we recover formulation
(\ref{eq:Robust}) from (\ref{eq:Robust-1}). The robust formulation
is based on the minimal excess value of $X$ over the value of the
benchmark prospect $Y$ for the whole set of choice functions specified
in $\mathcal{R}$. This kind of conservatism is to be used in decision
making for countering risks arising from ambiguity about the true
preferences. It is also very much in line with the philosophy of robust
optimization. In both \cite{armbruster2015decision} and \cite{delage2017minimizing},
the ``worst-case utility function'' and ``worst-case risk measure''
are considered in the same manner. This kind of framework is particularly
relevant in the context of group decision making whereby the least
favorable utility function from a member of the group is to be used
for the holistic decision making process.
\begin{rem}
In our upcoming development, one may omit the benchmark $Y$ and just
consider the function $\inf_{\rho\in\mathcal{R}}\rho\left(X\right)$
with only minor modification. We include the benchmark to stay consistent
with \cite{armbruster2015decision} and the wider literature on stochastic
dominance constrained optimization. The presence or absence of the
benchmark does not materially affect our main development. 
\end{rem}

The following proposition shows that the robust choice function $\psi\left(\cdot;\,\mathcal{R},\,Y\right)$
itself belongs to $\mathcal{R}_{iqv}$ whenever $\mathcal{R}\subset\mathcal{R}_{iqv}$.
\begin{prop}
\label{prop:robust} For any $\mathcal{R}\subset\mathcal{R}_{iqv}$
and $Y\in\mathcal{L}$, $\psi\left(\cdot;\,\mathcal{R},\,Y\right)$
is upper semi-continuous, increasing, and quasi-concave. 
\end{prop}

\begin{proof}
\textit{Upper semi-continuity:} Upper semi-continuity of a set of
functions is preserved by taking the point-wise infimum of the collection.

\textit{Monotonicity:} Monotonicity follows by 
\[
\psi\left(X;\,\mathcal{R},\,Y\right)=\inf_{\rho\in\mathcal{R}}\left\{ \rho\left(X\right)-\rho\left(Y\right)\right\} \leq\inf_{\rho\in\mathcal{R}}\left\{ \rho\left(Z\right)-\rho\left(Y\right)\right\} =\psi\left(Z;\,\mathcal{R},\,Y\right),
\]
since $\rho\left(X\right)\leq\rho\left(Z\right)$ for all $\rho\in\mathcal{R}$
whenever $X\leq Z$.

\textit{Quasi-concavity:} Quasi-concavity follows since 
\begin{align*}
\psi\left(\lambda\,X+\left(1-\lambda\right)Z;\,\mathcal{R},\,Y\right)=\, & \inf_{\rho\in\mathcal{R}}\left\{ \rho\left(\lambda\,X+\left(1-\lambda\right)Z\right)-\rho\left(Y\right)\right\} \\
\geq\, & \inf_{\rho\in\mathcal{R}}\left\{ \min\left\{ \rho\left(X\right),\,\rho\left(Z\right)\right\} -\rho\left(Y\right)\right\} \\
=\, & \inf_{\rho\in\mathcal{R}}\min\left\{ \rho\left(X\right)-\rho\left(Y\right),\,\rho\left(Z\right)-\rho\left(Y\right)\right\} \\
=\, & \min\left\{ \inf_{\rho\in\mathcal{R}}\left\{ \rho\left(X\right)-\rho\left(Y\right)\right\} ,\,\inf_{\rho\in\mathcal{R}}\left\{ \rho\left(Z\right)-\rho\left(Y\right)\right\} \right\} \\
=\, & \min\left\{ \psi\left(X;\,\mathcal{R},\,Y\right),\,\psi\left(Z;\,\mathcal{R},\,Y\right)\right\} ,
\end{align*}
for any $X,\,Z\in\mathcal{L}$ with $\lambda\in\left[0,\,1\right]$,
where the inequality follows by quasi-concavity of all $\rho\in\mathcal{R}$,
and the third equality follows by interchanging the order of minimization. 
\end{proof}
We now turn to discuss specification of the ambiguity set $\mathcal{R}$.
This set will have the following characteristics: 
\begin{itemize}
\item \textit{Preference elicitation:} For a sequence of pairs of prospects
$\{(W_{i},\,Y_{i})\}_{i\in\mathcal{I}}$, where $\mathcal{I}$ is
a finite index set, the decision maker prefers $W_{i}$ to $Y_{i}$
for all $i\in\mathcal{I}$. In this case, all admissible choice functions
in $\mathcal{R}_{iqv}$ that are consistent with the decision maker's
observed behaviors must satisfy $\rho\left(W_{i}\right)\geq\rho\left(Y_{i}\right)$
for all $i\in\mathcal{I}$. This form of preference elicitation also
appears in \cite{armbruster2015decision,delage2017minimizing}.
\item \textit{Normalization:} the decision maker's choice function satisfies
$\rho\left(0\right)=0$.
\item \textit{Lipschitz continuity:} the decision maker's choice function
$\rho$ is Lipschitz continuous. \textcolor{black}{Lipschitz continuity
ensures that the choice function does not vary too rapidly. Additionally,
this technical condition is necessary to apply a key representation
result for quasi-concave functions that appears in the next section.
Since $\psi\left(X;\,\alpha\,\mathcal{R},\,Y\right)=\alpha\,\psi\left(X;\,\mathcal{R},\,Y\right)$
for all $\alpha\geq0$ and any $\mathcal{R}\subset\mathcal{R}_{iqv}$,
we may specify the Lipschitz constant $L$ of $\rho$ arbitrarily.}
\end{itemize}
Our resulting specific ambiguity set is then 
\[
\mathcal{S}\triangleq\left\{ \rho\in\mathcal{R}_{iqv}\mbox{ : }\rho\left(W_{i}\right)\geq\rho\left(Y_{i}\right),\,\forall i\in\mathcal{I};\,\rho\left(0\right)=0;\,\rho\text{ is \ensuremath{L-}Lipschitz continuous}\right\} .
\]
We impose the $L-$Lipschitz condition since otherwise the set $\mathcal{S}$
is a cone, in which case $\psi\left(X;\,\mathcal{S},\,Y\right)$ may
not be finite-valued. In the next section we will develop a computational
recipe for evaluating $\psi\left(X;\,\mathcal{S},\,Y\right)$.

We conclude this section by introducing our robust choice model. Let
$\mathcal{Z}\subset\mathbb{R}^{m}$ be a set of available decisions
and let $G\mbox{ : }\mathcal{Z}\rightarrow\mathcal{L}$ be a random-variable-valued
mapping with realizations denoted $\left[G\left(z\right)\right]\left(\omega\right)$
for all $\omega\in\Omega$. The mapping $G$ captures the randomness
inherent in the underlying decision-making problem. In general, we
are interested in solving 
\begin{equation}
\max_{z\in\mathcal{Z}}\psi\left(G\left(z\right);\,\mathcal{S},\,Y\right)\triangleq\max_{z\in\mathcal{Z}}\inf_{\rho\in\mathcal{R}}\left\{ \rho\left(g(Z)\right)-\rho\left(Y\right)\right\} .\label{OPTIMIZATION}
\end{equation}
We make the following two key convexity assumptions on the problem
data $\mathcal{Z}$ and $G$. 
\begin{assumption}
\label{assu:convexity} (i) $\mathcal{Z}$ is closed and convex.

(ii) $G\mbox{ : }\mathcal{Z}\rightarrow\mathcal{L}$ is concave in
the sense that $\left[G\left(z\right)\right]\left(\omega\right)\mbox{ : }\mathcal{Z}\rightarrow\mathbb{R}^{n}$
is concave in $z\in\mathcal{Z}$ for $P-$almost all $\omega\in\Omega$. 
\end{assumption}

Assumption \ref{assu:convexity} ensures that our upcoming optimization
problems are convex, and the next proposition reveals the subsequent
key structural results of Problem (\ref{OPTIMIZATION}).
\begin{prop}
\label{prop:optimization} Problem (\ref{OPTIMIZATION}) is a quasi-concave
maximization problem. 
\end{prop}

\begin{proof}
Let $z_{1},\,z_{2}\in\mathcal{Z}$ and $\lambda\in\left[0,\,1\right]$.
For any $\rho\in\mathcal{R}_{iqv}$ we have 
\begin{align*}
\rho\left(G\left(\lambda\,z_{1}+\left(1-\lambda\right)z_{2}\right)\right)\geq\, & \rho\left(\lambda\,G\left(z_{1}\right)+\left(1-\lambda\right)G\left(z_{2}\right)\right)\\
\geq\, & \min\left\{ \rho\left(G\left(z_{1}\right)\right),\,\rho\left(G\left(z_{2}\right)\right)\right\} ,
\end{align*}
where the first inequality follows by monotonicity of $\rho$ and
concavity of $G$, and the second inequality follows by quasi-concavity
of $\rho$. The conclusion then follows by the previous part, the
fact that the infimum of quasi-concave functions is quasi-concave,
and the fact that the feasible region of Problem (\ref{OPTIMIZATION})
is convex. 
\end{proof}
We will return to Problem (\ref{OPTIMIZATION}) later in Section 6
after we carefully examine the function $\psi\left(X;\,\mathcal{S},\,Y\right)$.

\section{Support function representation}

In this section we turn to the primary issue of numerical \textit{evaluation}
of the robust choice function $\psi\left(X;\,\mathcal{S},\,Y\right)$.
At first glance, it is evident that $\psi\left(X;\,\mathcal{S},\,Y\right)$
cannot be evaluated with convex optimization because quasi-concavity
is not preserved under convex combination (and $\psi\left(X;\,\mathcal{S},\,Y\right)$
calls a minimization problem over a subset of quasi-concave functions).
In contrast, the robust choice functions in \cite{armbruster2015decision}
and \cite{delage2017minimizing} are amenable to convex optimization
because concave utility functions and convex risk measures are preserved
under convex combination. Despite this new difficulty, we can build
on the support function technique used in \cite{armbruster2015decision,delage2017minimizing}
and augment it for our present setting in $\mathcal{R}_{iqv}$.

We begin with some preliminary definitions and facts associated with
support functions. For emphasis, in the following definition and throughout
we consider \textit{upper} support functions which dominate a target
function from above rather than below (as in the convex and quasi-convex
cases). We remind the reader that $\langle a,\,b\rangle$ denotes
the Euclidean inner product.
\begin{defn}
Let $f\text{ : }\mathbb{R}^{d}\rightarrow\mathbb{R}$. Recall that
$f$ is said to be \textit{majorized} by a function $g$ if 
\[
f\left(x\right)\leq g\left(x\right),\,\forall x\in\text{dom}\,f,
\]
and $g$ is an upper \textit{support function} of $f$ at $x\in\mathbb{R}^{d}$
if $f\left(x\right)=g\left(x\right)$. Here and later on, $\text{dom}\,f$
denotes the domain of $f$. A vector $s\in\mathbb{R}^{d}$ is called
a \textit{subgradient} of $f$ at $x\in\mathbb{R}^{d}$ if 
\[
f\left(y\right)\leq f\left(x\right)+\langle s,\,y-x\rangle,\,\forall y\in\text{dom}\,f.
\]
We denote the set of subgradients of $f$ at $x$ by $\partial f\left(x\right)$
and call the latter {\em subdifferential}. A vector $s\in\mathbb{R}^{d}$
is an \textit{upper subgradient} of $f$ at $x\in\mathbb{R}^{d}$
if 
\[
f\left(y\right)\leq f\left(x\right)+\langle s,\,y-x\rangle,\,\forall y\in\left\{ y\in\text{dom}\,f\text{ : }f\left(y\right)\geq f\left(x\right)\right\} .
\]
We denote the subdifferential of $f$ at $x$ by $\partial^{+}f\left(x\right)$. 
\end{defn}

When $f$ is concave and subdifferentiable at $x$, i.e., $\partial f(x)\neq\emptyset$,
the linear function 
\[
l\left(y\right)=f\left(x\right)+\langle a,\,y-x\rangle
\]
is a support function of $f$ at $x$ for any $a\in\partial f\left(x\right)$.
The following theorem gives rise to a characterization of concave
functions by their support functions.
\begin{thm}
\label{thm:Support-concave} Let $f\text{ : }\mathbb{R}^{d}\to\mathbb{R}$.
The following assertions hold. 
\begin{itemize}
\item[(i)] $f$ is a concave function if and only if there exists an index set
$\mathcal{J}$ such that 
\begin{equation}
f\left(x\right)=\inf_{j\in\mathcal{J}}l_{j}\left(x\right),\,\forall x\in\text{dom}\,f,\label{eq:Support-concave}
\end{equation}
where $\mathcal{J}$ is possibly infinite and $l_{j}\left(x\right)=\langle a_{j},\,x\rangle+b_{j}$
for all $j\in\mathcal{J}$.
\item[(ii)] For any finite set $\Theta\subset\mathbb{R}^{d}$ and values $\left\{ v_{\theta}\right\} _{\theta\in\Theta}\subset\mathbb{R}$,
$\hat{f}\text{ : }\mathbb{R}^{d}\rightarrow\mathbb{R}$ defined by
\[
\hat{f}\left(x\right)=\min_{a,\,b}\left\{ \langle a,\,x\rangle+b\text{ : }\langle a,\,\theta\rangle+b\geq v_{\theta},\,\forall\theta\in\Theta\right\} 
\]
is concave. Moreover, for all concave functions $\tilde{f}$ with
$\tilde{f}\left(\theta\right)\geq v_{\theta}$, $\hat{f}\leq\tilde{f}$. 
\end{itemize}
\end{thm}

\begin{proof}
Part (i). These results are well known, see for instance \cite[Section 3.2]{Boyd2004}
and \cite{Rockafellar1970}.

Part (ii). It is immediate that $\hat{f}$ as defined is concave,
as it is of the form (\ref{eq:Support-concave}). Note that the hypograph
of $\hat{f}$ is the convex hull of $\left\{ \theta,\,v_{\theta}\right\} _{\theta\in\Theta}$.
The convex hull of a set of points is by definition the smallest convex
set containing these points. See also \cite[Section 6.5.5]{Boyd2004}. 
\end{proof}
Theorem \ref{thm:Support-concave} says that one can recover concave
functions by taking the infimum of their support functions. Moreover,
it shows that support functions can be used for the construction of
the ``lowest'' concave function that dominates a fixed set of values.
The results provide the basis for tractable formulations of PRO models
in \cite{armbruster2015decision,Haskell_Aspects_2015,delage2017minimizing}.
For further details on other applications of this result, we refer
the reader to \cite[Section 6.5.5]{Boyd2004} for a discussion of
interpolation with convex functions.

When $f$ is quasi-concave and upper subdifferentiable at $x$, the
piecewise linear function 
\[
h\left(y\right)=\max\left\{ f\left(x\right)+\langle a,\,y-x\rangle,\,f\left(x\right)\right\} 
\]
is a support function of $f$ at $x$ for any $a\in\partial^{+}f\left(x\right)$.
Note that these functions $h$ are the maximum of two linear functions
and so open upwards. We informally refer to such $h$ as ``hockey
stick'' functions in recognition of this shape. Note that these functions
are quasi-concave themselves.

For any $a\in\mathbb{R}^{d}$ and $b,\,c\in\mathbb{R}$, the function
$h\text{ : }\mathbb{R}^{d}\rightarrow\mathbb{R}$ defined by 
\[
h\left(x\right)=\max\left\{ \langle a,\,x\rangle+b,\,c\right\} 
\]
is quasi-concave since its upper level sets are convex. In particular,
it is easy to verify that for $t\leq c$,
\[
\left\{ x\in\mathbb{R}^{d}\text{ : }h\left(x\right)\geq t\right\} =\mathbb{R}^{d}
\]
and for $t>c$,
\[
\left\{ x\in\mathbb{R}^{d}\text{ : }h\left(x\right)\geq t\right\} =\left\{ x\in\mathbb{R}^{d}\text{ : }\langle a,\,x\rangle+b\geq t\right\} 
\]
which is convex (it is a half-space) since it is the upper level set
of an affine function.

The following result is the analog of Theorem \ref{thm:Support-concave}
for quasi-concave functions, it characterizes quasi-concave functions
via this class of ``hockey stick'' support functions. The result
will provide a theoretical foundation for tractable reformulation
of our quasi-concave robust choice function. We note that the first
part of the following theorem requires the stronger assumption that
our function of interest $f$ is $L-$Lipschitz continuous.
\begin{thm}
\label{thm:Support-quasiconcave} Let $f\text{ : }\mathbb{R}^{d}\to\mathbb{R}$.
The following assertions hold. 
\begin{itemize}
\item[(i)] Suppose that $f$ is quasi-concave and $L-$Lipschitz continuous.
Then 
\begin{equation}
f\left(x\right)=\inf_{j\in\mathcal{J}}h_{j}\left(x\right),\,\forall x\in\text{dom}\,f,\label{eq:Support-quasiconcave}
\end{equation}
where $\mathcal{J}$ is possibly infinite and $h_{j}\left(x\right)=\max\left\{ \langle a_{j},\,x\rangle+b_{j},\,c_{j}\right\} $
for $\|a_{j}\|_{2}\leq L$, $j\in\mathcal{J}$.
\item[(ii)] If $f$ has a representation (\ref{eq:Support-quasiconcave}), then
it is quasi-concave.
\item[(iii)] For any finite set $\Theta\subset\mathbb{R}^{d}$ and values $\left\{ v_{\theta}\right\} _{\theta\in\Theta}\subset\mathbb{R}$,
$\hat{f}\text{ : }\mathbb{R}^{d}\rightarrow\mathbb{R}$ defined by
\begin{align*}
\hat{f}\left(x\right)\triangleq\inf_{a,\,b,\,c}\, & \max\left\{ \langle a,\,x\rangle+b,\,c\right\} \\
\text{s.t.}\, & \max\left\{ \langle a,\,\theta\rangle+b,\,c\right\} \geq v_{\theta},\,\forall\theta\in\Theta,
\end{align*}
is quasi-concave. Furthermore, the graph of $\hat{f}$ is the quasi-concave
envelope for the set of points $\left\{ \left(\theta,\,v_{\theta}\right)\text{ : }\theta\in\Theta\right\} $. 
\end{itemize}
\end{thm}

\begin{proof}
Part (i). By \cite[Theorem 2.3]{plastria1985lower}, any Lipschitz
continuous quasi-concave function is upper subdifferentiable on its
domain. Thus $\partial^{+}f\left(x\right)\neq\emptyset$ for any $x\in\text{dom}\,f$
under the assumption that $f$ is L-Lipschitz, and for any $s_{x}\in\partial^{+}f\left(x\right)$,
$\|s_{x}\|_{2}\leq L$. Moreover, for any $x\in\text{dom}\,f$ and
$s_{x}\in\partial^{+}f\left(x\right)$, $f$ is supported by $h_{x}\left(y\right)=\max\left\{ f\left(x\right)+\langle s_{x},\,y-x\rangle,\,f\left(x\right)\right\} $
at $x$. By taking the infimum of all support functions defined as
such, we have
\[
f\left(y\right)\leq\inf_{x\in\text{dom}\,f,\,s_{x}\in\partial f\left(x\right)}\max\left\{ f\left(x\right)+\langle s_{x},\,y-x\rangle,\,f\left(x\right)\right\} ,\,\forall y\in\text{dom}\,f.
\]
Furthermore, for any $y\in\text{dom}\,f$,
\[
f\left(y\right)=h_{y}\left(y\right)\geq\inf_{x\in\text{dom}\,f}h_{x}\left(y\right)
\]
since $f\left(y\right)=\max\left\{ f\left(y\right)+\langle s_{y},\,y-y\rangle,\,f\left(y\right)\right\} =h_{y}\left(y\right)$.

Part (ii). For any $t\in\mathbb{R}$ we have 
\[
\left\{ x\in\mathbb{R}^{d}\text{ : }f\left(x\right)\geq t\right\} =\left\{ x\in\mathbb{R}^{d}\text{ : }\inf_{j\in\mathcal{J}}h_{j}\left(x\right)\geq t\right\} =\bigcap_{j\in\mathcal{J}}\left\{ x\in\mathbb{R}^{d}\text{ : }h_{j}\left(x\right)\geq t\right\} ,
\]
each $\left\{ x\in\mathbb{R}^{d}\text{ : }h_{j}\left(x\right)\geq t\right\} $
is convex by quasi-concavity of $h_{j}$ for all $j\in\mathcal{J}$,
and the intersection of convex sets is convex.

Part (iii). The quasi-concavity of $\hat{f}$ follows from the fact
that each function $\max\left\{ \langle a,\,x\rangle+b,\,c\right\} $
is quasi-concave and the infimum of such functions is also quasi-concave.
In what follows, we prove the second part of the statement. For any
$t\in\mathbb{R}$, each $\theta\in\Theta$ with $v_{\theta}\geq t$
belongs to $\left\{ x\in\mathbb{R}^{d}\text{ : }\hat{f}\left(x\right)\geq t\right\} $
because by definition $\hat{f}\left(\theta\right)\geq v_{\theta}\geq t$.
Since $\hat{f}$ is quasi-concave and has convex upper level sets,
we must then have 
\begin{equation}
\text{conv}\left\{ \theta\in\Theta\text{ : }v_{\theta}\geq t\right\} \subseteq\left\{ x\in\mathbb{R}^{d}\text{ : }\hat{f}\left(x\right)\geq t\right\} ,\label{eq:Thm4.3-proof-1}
\end{equation}
where ``conv'' denotes the convex hull of a set.

For $x\notin\text{conv}\left\{ \theta\in\Theta\text{ : }v_{\theta}\geq t\right\} $,
we may choose appropriate parameters $a\in\mathbb{R}^{d}$, $b\in\mathbb{R}$
and $c\in\mathbb{R}$ such that $\max\left\{ \langle a,\,\theta\rangle+b,\,c\right\} \geq v_{\theta}$
for all $\theta\in\Theta$ and $\max\left\{ \langle a,\,x\rangle+b,\,c\right\} <t$.
To construct such a hockey-stick function, first let $y\in\mathbb{R}^{d}$
be the projection of $x$ onto the convex set $\text{conv}\left\{ \theta\in\Theta\text{ : }v_{\theta}\geq t\right\} $.
Next, by virtue of the separation theorem in convex analysis, there
exists $\tilde{a}\in\mathbb{R}^{d}$ such that $\langle\tilde{a},\,y-x\rangle\geq\max_{\theta\in\Theta}v_{\theta}$
and $\langle\tilde{a},\,\theta-y\rangle\geq0$ for all $\theta\in\left\{ \theta\in\Theta\text{ : }v_{\theta}\geq t\right\} $.
Let $\tilde{c}=\max_{\theta\in\Theta}\left\{ v_{\theta}\text{ : }v_{\theta}<t\right\} <t$
and $\tilde{b}=-\langle\tilde{a},\,x\rangle$. Then
\begin{eqnarray*}
\max\left\{ \langle\tilde{a},\,\theta\rangle+\tilde{b},\,\tilde{c}\right\}  & = & \max\left\{ \langle\tilde{a},\,\theta-x\rangle,\,\tilde{c}\right\} \\
 & = & \max\left\{ \langle\tilde{a},\,\theta-y\rangle+\langle\tilde{a},\,y-x\rangle,\,\tilde{c}\right\} \\
 & \geq & \left\{ \begin{array}{ll}
\max\left\{ \langle\tilde{a},\,y-x\rangle,\,\tilde{c}\right\} , & \text{for}\;\theta\in\left\{ \theta\in\Theta\text{ : }v_{\theta}\geq t\right\} \\
\tilde{c}, & \text{in any case}
\end{array}\right.\\
 & \geq & v_{\theta},\forall\theta\in\Theta.
\end{eqnarray*}
By the definition of $\hat{f}(x)$, the inequality above implies that
$\hat{f}\left(x\right)\leq\max\left\{ \langle\tilde{a},\,x\rangle+\tilde{b},\,\tilde{c}\right\} $.
On the other hand, it is easy to verify that $\max\left\{ \langle\tilde{a},\,x\rangle+\tilde{b},\,\tilde{c}\right\} =\tilde{c}<t$.
Thus, we arrive at $\hat{f}\left(x\right)<t$, which enables us to
deduce that $x\notin\left\{ x\in\mathbb{R}^{d}\text{ : }\hat{f}\left(x\right)\geq t\right\} $
and subsequently
\begin{equation}
\left\{ x\in\mathbb{R}^{d}\text{ : }\hat{f}\left(x\right)\geq t\right\} \subseteq\text{conv}\left\{ \theta\in\Theta\text{ : }v_{\theta}\geq t\right\} .\label{eq:Thm4.3-proof-2}
\end{equation}
A combination of (\ref{eq:Thm4.3-proof-1}) and (\ref{eq:Thm4.3-proof-2})
yields
\[
\text{conv}\left\{ \theta\in\Theta\text{ : }v_{\theta}\geq t\right\} =\left\{ x\in\mathbb{R}^{d}\text{ : }\hat{f}\left(x\right)\geq t\right\} .
\]
Since the convex hull is the smallest convex set containing a set
of points, it must be that 
\[
\left\{ x\in\mathbb{R}^{d}\text{ : }\hat{f}\left(x\right)\geq t\right\} =\text{conv}\left\{ \theta\in\Theta\text{ : }v_{\theta}\geq t\right\} \subseteq\left\{ x\in\mathbb{R}^{d}\text{ : }\tilde{f}\left(x\right)\geq t\right\} 
\]
for any other quasi-concave $\tilde{f}$ with $\tilde{f}\left(\theta\right)\geq v_{\theta}$
for all $\theta\in\Theta$. This same reasoning holds for all $t\in\mathbb{R}$,
so $\left\{ x\in\mathbb{R}^{d}\text{ : }\hat{f}\left(x\right)\geq t\right\} \subseteq\left\{ x\in\mathbb{R}^{d}\text{ : }\tilde{f}\left(x\right)\geq t\right\} $
for all $t\in\mathbb{R}$ and $\hat{f}\leq\tilde{f}$. 
\end{proof}
We note that we may use any norm (not necessarily the Euclidean norm)
to enforce Lipschitz continuity in part (i) of Theorem \ref{thm:Support-quasiconcave}
since all norms on $\mathbb{R}^{d}$ are equivalent. In parallel to
Theorem \ref{thm:Support-concave}, Theorem \ref{thm:Support-quasiconcave}
gives conditions where a quasi-concave function can be recovered by
taking the infimum of its support functions (which are hockey stick
functions in this case). Moreover, Theorems \ref{thm:Support-concave}
and \ref{thm:Support-quasiconcave} give conditions for constructing
the ``lowest'' quasi-concave function that contains a fixed set
of values.

\subsection{Reformulation as a mixed-integer linear program}

Theorem \ref{thm:Support-quasiconcave} gives an explicit form for
the ``worst-case'' quasi-concave function that dominates a set of
values $\left\{ v_{\theta}\right\} _{\theta\in\Theta}$ over a finite
set $\Theta$. In fact, this is exactly what we need to derive a tractable
reformulation of $\psi\left(X;\,\mathcal{S},\,Y\right)$. The remaining
challenge is to put the correct conditions on the values $\left\{ v_{\theta}\right\} _{\theta\in\Theta}$,
which will take the form of an optimization problem. To this effect,
for the remainder of this section we introduce the major assumption
that the underlying sample space is finite. 
\begin{assumption}
\label{assu:finite} The sample space $\Omega$ is finite. 
\end{assumption}

Assumption \ref{assu:finite} also appears in \cite{armbruster2015decision,Haskell_Aspects_2015,delage2017minimizing}
where it is used for obtaining tractable optimization formulations.
However, in the case when $\Omega$ is continuous, it is possible
to develop a discrete approximation (see \cite[Section 5]{delage2017shortfall}
and \cite{Haskell_Aspects_2015}).

Under Assumption \ref{assu:finite}, we adopt the convention that
a prospect $X\in\mathcal{L}$ may be identified with the vector of
its realizations 
\[
\vec{X}=\left(X\left(\omega\right)\right)_{\omega\in\Omega}\in\mathbb{R}^{n\,|\Omega|}.
\]
This convention first appeared in \cite{delage2017minimizing} and
depends on a finite sample space $\Omega$. In this way, there is
a one-to-one correspondence between elements of $\mathcal{L}$ and
$\mathbb{R}^{n\,|\Omega|}$. We now define 
\[
\Theta\triangleq\left\{ \vec{0}\right\} \cup\left\{ \vec{W}_{i}\right\} _{i\in\mathcal{I}}\cup\left\{ \vec{Y}_{i}\right\} _{i\in\mathcal{I}}\cup\left\{ \vec{Y}\right\} 
\]
to be the union of all the prospects used in the definition of $\mathcal{S}$
(including the constant prospect $\vec{0}$ which is used in the normalization
condition) along with the benchmark $Y$.

To proceed on, we let 
\[
S\left(v\right)\triangleq\left\{ \rho\text{ : }\mathcal{L}\rightarrow\mathbb{R}\text{ s.t. }\rho\left(\theta\right)=v_{\theta},\,\forall\theta\in\Theta,\,\rho\text{ is law invariant}\right\} 
\]
denote the set of choice functions that take the values $\left\{ v_{\theta}\right\} _{\theta\in\Theta}$
on the set $\Theta$. We now evaluate the worst-case choice function
$\psi\left(X;\,\mathcal{S},\,Y\right)$ with the following procedure: 
\begin{enumerate}
\item Set the values $\left\{ v_{\theta}\right\} _{\theta\in\Theta}$ of
$\rho$ on $\Theta$. We only consider these values at first because
they give sufficient information to construct $\rho$ on the rest
of $\mathcal{L}$, as we will show. 
\begin{enumerate}
\item The values $\left\{ v_{\theta}\right\} _{\theta\in\Theta}$ must satisfy
the majorization characterization for quasi-concave functions (a quasi-concave
function is majorized by its hockey stick support function at every
point on its graph). This is equivalent to determining if $S\left(v\right)\cap\mathcal{R}_{iqv}\ne\emptyset$. 
\item The values $\left\{ v_{\theta}\right\} _{\theta\in\Theta}$ must satisfy
the preference elicitation condition given in the definition of $\mathcal{S}$. 
\item The values $\left\{ v_{\theta}\right\} _{\theta\in\Theta}$ must satisfy
the $L-$Lipschitz continuity condition given in the definition of
$\mathcal{S}$. 
\end{enumerate}
\item Once the values of $\rho$ are fixed on $\Theta$ satisfying the above
conditions, interpolate using hockey stick support functions to determine
the value at any $X\in\mathcal{L}$. 
\end{enumerate}
This procedure results in the following optimization problem: 
\begin{align}
\min_{a,\,s,\,b,\,c,\,v}\, & \max\left\{ \langle a,\,\vec{X}\rangle+b,\,c\right\} -v_{\vec{Y}}\label{ROBUST}\\
\text{s.t.}\, & \max\left\{ v_{\theta}+\langle s_{\theta},\,\theta'-\theta\rangle,\,v_{\theta}\right\} \geq v_{\theta'}, & \forall\theta\ne\theta';\,\theta,\,\theta'\in\Theta,\label{ROBUST-1}\\
 & v_{\vec{0}}=0,\label{ROBUST-2}\\
 & v_{\vec{W}_{i}}\geq v_{\vec{Y}_{i}}, & \forall i\in\mathcal{I},\label{ROBUST-3}\\
 & s_{\theta}\geq0,\,\|s_{\theta}\|_{\infty}\leq L, & \forall\theta\in\Theta,\label{ROBUST-4}\\
 & \max\left(\langle a,\,\theta\rangle+b,\,c\right)\geq v_{\theta}, & \forall\theta\in\Theta,\label{ROBUST-5}\\
 & a\geq0,\,\|a\|_{\infty}\leq L.\label{ROBUST-6}
\end{align}
Intuitively, constraint (\ref{ROBUST-1}) is the majorization characterization
for the values $v_{\theta}=\rho\left(\theta\right)$ for all $\theta\in\Theta$;
constraint (\ref{ROBUST-2}) is the normalization constraint; constraint
(\ref{ROBUST-3}) corresponds to the preference elicitation requirement
in the definition of $\mathcal{S}$; constraint (\ref{ROBUST-4})
requires the support functions used to characterize $\left\{ v_{\theta}\right\} _{\theta\in\Theta}$
to be increasing and Lipschitz continuous; constraint (\ref{ROBUST-5})
requires the support function used to determine the value of $\rho$
at $X$ to majorize $\rho$; and finally constraint (\ref{ROBUST-6})
requires the support function used to determine the value of $\rho$
at $X$ to be increasing and Lipschitz continuous.

Problem (\ref{ROBUST}) - (\ref{ROBUST-6}) has several features in
common with \cite{armbruster2015decision,delage2017minimizing}. In
particular, all of these formulations have a support function characterization
that ensures convexity/concavity/quasi-concavity, and all of these
formulations have constraints corresponding to preference elicitation.
The main difference is that the formulations in \cite{armbruster2015decision,delage2017minimizing}
are based on linear functions while our formulation is built on hockey
stick functions.
\begin{rem}
The convention $\vec{X}=\left(X\left(\omega\right)\right)_{\omega\in\Omega}$
avoids the requirement that the convex hull of $\Theta$ contain the
support of $X$ in Problem (\ref{ROBUST}) - (\ref{ROBUST-6}) as
in \cite{armbruster2015decision,delage2017shortfall}. 
\end{rem}

The next theorem formally verifies the correctness of this formulation.
\begin{thm}
\label{thm:ROBUST} Suppose Assumption \ref{assu:finite} holds. Given
$X\in\mathcal{L}$, the optimal value of Problem (\ref{ROBUST}) -
(\ref{ROBUST-6}) is equal to $\psi\left(X;\,\mathcal{S},\,Y\right)$. 
\end{thm}

\begin{proof}
To begin, we may partition the set of choice functions by their values
on the finite set $\Theta$. We have the equivalence 
\begin{align*}
\psi\left(X;\,\mathcal{S},\,Y\right)=\, & \inf_{\rho\in\mathcal{S}}\left\{ \rho\left(X\right)-\rho\left(Y\right)\right\} \\
=\, & \min_{v\in\mathbb{R}^{|\Theta|}}\inf_{\rho\in S\left(v\right)\cap\mathcal{S}}\left\{ \rho\left(X\right)-\rho\left(Y\right)\right\} \\
=\, & \min_{v\in\mathbb{R}^{|\Theta|}}\left\{ \psi\left(X;\,S\left(v\right)\cap\mathcal{S},\,Y\right)\text{ : }S\left(v\right)\cap\mathcal{S}\ne\emptyset\right\} ,
\end{align*}
where we use $\mathcal{S}=\cup_{v\in\mathbb{R}^{|\Theta|}}\left\{ S\left(v\right)\cap\mathcal{S}\right\} $,
and $|\Theta|$ denotes the cardinality of $\Theta$. This condition
is simply saying that $\psi\left(X;\,\mathcal{S},\,Y\right)$ can
be understood as either minimizing over $\mathcal{S}$ directly, or
first fixing the relevant values $\left\{ v_{\theta}\right\} _{\theta\in\Theta}$
on $\Theta$ and then minimizing over functions in $\mathcal{S}$
that coincide with those values.

Next define the sets
\[
S_{N}\triangleq\left\{ \rho\text{ : }\mathcal{L}\rightarrow\mathbb{R}\text{ : }\rho\left(0\right)=0\right\} ,
\]

\[
S_{E}\triangleq\left\{ \rho\text{ : }\mathcal{L}\rightarrow\mathbb{R}\text{ s.t. }\rho\left(W_{i}\right)\geq\rho\left(Y_{i}\right),\,\forall i\in\mathcal{I}\right\} ,
\]
and 
\[
S_{L}\triangleq\left\{ \rho\text{ : }\mathcal{L}\rightarrow\mathbb{R}\text{ s.t. }\rho\text{ is \ensuremath{L-}Lipschitz continuous}\right\} .
\]
Then, using the fact that $\mathcal{S}=S_{N}\cap S_{E}\cap S_{L}\cap\mathcal{R}_{iqv}$,
we have 
\begin{align*}
\psi\left(X;\,\mathcal{S},\,Y\right)=\min_{v\in\mathbb{R}^{|\Theta|}}\, & \psi\left(X;\,S\left(v\right)\cap\mathcal{S},\,Y\right)\\
\text{s.t.}\, & S\left(v\right)\subseteq S_{N},\,S\left(v\right)\subseteq S_{E},\,S\left(v\right)\subseteq S_{L},\,S\left(v\right)\cap\mathcal{R}_{iqv}\ne\emptyset.
\end{align*}
Constraint $S\left(v\right)\subseteq S_{E}$ is just the preference
elicitation condition (\ref{ROBUST-3}) and constraint $S\left(v\right)\subseteq S_{L}$
is just the Lipschitz continuity requirement (\ref{ROBUST-4}). Both
of these conditions only constrain the values of $\rho$ on $\Theta$
via $v$. Constraint $S\left(v\right)\cap\mathcal{R}_{iqv}\ne\emptyset$
states that there must exist a function in $\mathcal{R}_{iqv}$ that
takes the values $\left\{ v_{\theta}\right\} _{\theta\in\Theta}$
on $\Theta$. This requirement is enforced by: (i) constraint (\ref{ROBUST-1}),
the majorization characterization of the quasi-concavity of $\left\{ v_{\theta}\right\} _{\theta\in\Theta}$;
and (ii) constraint (\ref{ROBUST-4}), the requirement that the support
functions that majorize $\rho$ on $\Theta$ be increasing and Lipschitz
continuous.

So, it remains to evaluate $\psi\left(X;\,S\left(v\right)\cap\mathcal{S},\,Y\right)$
for fixed $v\in\mathbb{R}^{|\Theta|}$, which is explicitly 
\[
\inf_{\rho\in S\left(v\right)\cap\mathcal{S}}\left\{ \rho\left(X\right)-\rho\left(Y\right)\right\} =\inf_{\rho\in S\left(v\right)\cap\mathcal{S}}\rho\left(X\right)-\rho\left(Y\right)
\]
since the value $\rho\left(Y\right)$ is fixed for $\rho\in S\left(v\right)\cap\mathcal{S}$
by construction of $S\left(v\right)$. Over all increasing quasi-concave
functions in $S\left(v\right)$, the one minimizing $\inf_{\rho\in S\left(v\right)\cap\mathcal{S}}\rho\left(X\right)$
attains the value 
\[
\min\left\{ \max\left\{ \langle a,\,\vec{X}\rangle+b,\,c\right\} \text{ : }\max\left\{ \langle a,\,\theta\rangle+b,\,c\right\} \geq v_{\theta},\,\forall\theta\in\Theta\right\} ,
\]
which is captured by constraint (\ref{ROBUST-5}) that requires the
support function $\max\left(\langle a,\,\vec{X}\rangle+b,\,c\right)$
to majorize $\left\{ v_{\theta}\right\} _{\theta\in\Theta}$. 
\end{proof}
Problem (\ref{ROBUST}) - (\ref{ROBUST-6}) is finite-dimensional,
but it is not a linear programming problem (or even a convex optimization
problem) due to constraints (\ref{ROBUST-1}) and (\ref{ROBUST-5})
which require a convex function to be \textit{greater} than a linear
term. However, it is possible to transform Problem (\ref{ROBUST})
- (\ref{ROBUST-6}) into a mixed-integer linear programming problem
(MILP) using standard techniques. We obtain the following MILP where
we use a new constant $M\gg0$:

\begin{align}
\min_{a,\,s,\,b,\,c,\,v,\,x,\,y}\, & t-v_{\vec{Y}}\label{MILP}\\
\text{s.t.}\, & t\geq\langle a,\,\vec{X}\rangle+b,\label{MILP-1}\\
 & t\geq c,\label{MILP-2}\\
 & \langle s_{\theta},\,\theta'-\theta\rangle+v_{\theta}+M\,x_{\theta,\,\theta'}\geq v_{\theta'}, & \forall\theta\ne\theta';\,\theta,\,\theta'\in\Theta,\label{MILP-3}\\
 & v_{\theta}+M\left(1-x_{\theta,\,\theta'}\right)\geq v_{\theta'}, & \forall\theta\ne\theta';\,\theta,\,\theta'\in\Theta,\label{MILP-4}\\
 & v_{\vec{0}}=0,\label{MILP-5}\\
 & v_{\vec{W}_{i}}\geq v_{\vec{Y}_{i}}, & \forall i\in\mathcal{I},\label{MILP-6}\\
 & s_{\theta}\geq0,\,\|s_{\theta}\|_{\infty}\leq L, & \forall\theta\in\Theta,\label{MILP-7}\\
 & \langle a,\,\theta\rangle+b+M\,y_{\theta}\geq v_{\theta}, & \forall\theta\in\Theta,\label{MILP-8}\\
 & c+M\left(1-y_{\theta}\right)\geq v_{\theta}, & \forall\theta\in\Theta,\label{MILP-9}\\
 & a\geq0,\,\|a\|_{\infty}\leq L,\label{MILP-10}\\
 & x_{\theta,\,\theta'}\in\left\{ 0,\,1\right\} , & \forall\theta\ne\theta';\,\theta,\,\theta'\in\Theta,\nonumber \\
 & y_{\theta}\in\left\{ 0,\,1\right\} , & \forall\theta\in\Theta.\nonumber 
\end{align}
In explanation, constraints (\ref{MILP-1}) and (\ref{MILP-2}) replace
the term $\max\left(\langle a,\,\vec{X}\rangle+b,\,c\right)$ that
appears in the objective of Problem (\ref{ROBUST}) - (\ref{ROBUST-6})
with linear terms via the epigraphical transformation. Constraints
(\ref{MILP-3}) and (\ref{MILP-4}) along with the binary constraints
on $x$ replace constraints (\ref{ROBUST-1}) with disjunctive constraints;
likewise, constraints (\ref{MILP-8}) and (\ref{MILP-9}) along with
the binary constraints on $y$ replace constraints (\ref{ROBUST-5})
with disjunctive constraints.
\begin{rem}
In terms of computation, Problem (\ref{MILP}) - (\ref{MILP-10})
can be effectively solved by Bender's decomposition as proposed in
\cite{Christensen2008,Rahmaniani2017,Taskin2010}. In \cite{Codato2006},
the authors explain that the linear programming relaxation of the
MILP is typically a poor approximation due to the big-$M$ coefficients.
In fact, the binary solutions of LP relaxations are only marginally
affected by the addition of continuous variables and the associated
constraints. Different choices of $M$ will affect the branching process
in the algorithm, as well as the constraints in the LP relaxation.
In recognition of this difficulty, these authors introduce ``Combinatorial
Bender's cuts'' that can generate more effective cuts for the master
problem and that also avoid the difficulty of choosing the constant
$M$. Alternatively, based on \cite{Khurana2005}, we could use a
convex hull reformulation for each disjunction in Problem (\ref{MILP})
- (\ref{MILP-10}). This reformulation will always give a tighter
bound than the big-$M$ formulation, at the expense of a large number
of variables and constraints. Both of these approaches converge to
the global optimum after finitely many iterations, and can improve
upon basic methods for solving Problem (\ref{MILP}) - (\ref{MILP-10}). 
\end{rem}

\subsection{The worst-case choice function}

We conclude this section by giving further details on the explicit
form of the worst-case choice function. In \cite{armbruster2015decision},
the worst-case utility function is shown to be piecewise linear concave
by using a support function argument. This is possible because the
worst-case utility function in question is a mapping from $\mathbb{R}$
to $\mathbb{R}$. Our present setting is more complicated because:
(i) we deal with multi-attribute prospects and (ii) we are concerned
with quasi-concave functions. Yet, as we will see shortly, our worst-case
choice function can also be constructed explicitly. The Delaunay triangulation,
defined next, is the key to this approach. 
\begin{defn}
\label{def:Delaunay} (i) A simplex $\sigma$ is a polytope in $\mathbb{R}^{d}$
such that $\sigma$ is the convex hull of $d+1$ affinely independent
points.

(ii) A simplicial complex $\mathcal{C}$ is a finite collection of
simplices such that: $\forall\sigma\in\mathcal{C}$, $\sigma$ is
a simplex; $\sigma_{1}\in\mathcal{C}$ and $\sigma_{1}\subset\sigma_{2}$
imply $\sigma_{2}\in\mathcal{C}$; and for any $\sigma_{1},\,\sigma_{2}\in\mathcal{C}$,
either $\sigma_{1}\cap\sigma_{2}=\emptyset$ or $\sigma_{1}\cap\sigma_{2}\in\mathcal{C}$.

(iii) A triangulation $\text{T}\left(\mathcal{X}\right)$ of a finite
set $\mathcal{X}\subset\mathbb{R}^{d}$ is a simplicial complex whose
vertices belong to $\mathcal{X}$ and whose union is the convex hull
of $\mathcal{X}$.

(iv) A circumsphere (circumscribed sphere) of a simplex $\sigma$
is a sphere that contains all of the vertices of $\sigma$.

(v) A Delaunay triangulation of a finite set $\mathcal{X}\subset\mathbb{R}^{d}$
is a triangulation $\text{DT}\left(\mathcal{X}\right)$ such that
no point in $\mathcal{X}$ is inside the circumsphere of any simplex
in $\text{DT}\left(\mathcal{X}\right)$. 
\end{defn}

We use the Delaunay triangulation $\text{DT}\left(\Theta\right)$
of $\Theta$ in our construction of the worst-case choice function.
The next theorem shows that the worst-case choice function (without
Lipschitz continuity) is piecewise constant, and follows directly
from Theorem \ref{thm:Support-quasiconcave}(iii). 
\begin{thm}
Choose $v\in\mathbb{R}^{|\Theta|}$ such that $S\left(v\right)\cap\mathcal{R}_{iqv}\ne\emptyset$,
then $\hat{\rho}\left(\cdot\right)\triangleq\inf_{\rho\in S\left(v\right)\cap\mathcal{R}_{iqv}}\rho\left(\cdot\right)$
is given by 
\[
\hat{\rho}\left(X\right)=\max\left\{ \min_{\theta\in\sigma\cap\Theta}v_{\theta}\text{ : }\vec{X}\in\sigma,\,\sigma\in\text{DT}\left(\Theta\right)\right\} .
\]
\end{thm}

Next we consider the case where Lipschitz continuity is enforced.
This result also follows directly from Theorem \ref{thm:Support-quasiconcave}(iii),
only now the worst-case choice function is piecewise linear.
\begin{thm}
Choose $v\in\mathbb{R}^{|\Theta|}$ such that $S\left(v\right)\cap\mathcal{R}_{iqv}\ne\emptyset$,
then $\hat{\rho}\left(\cdot\right)\triangleq\inf_{\rho\in S\left(v\right)\cap\mathcal{R}_{iqv}}\rho\left(\cdot\right)$
is given by 
\[
\hat{\rho}\left(X\right)=\left\{
\langle\alpha_{\sigma},\,\vec{X}\rangle+\beta_{\sigma}, \vec{X}\in\sigma,\,\sigma\in\text{DT}\left(\Theta\right)\right\},
\]
where $\langle\alpha_{\sigma},\,\theta\rangle+\beta_{\sigma}=v_{\theta}$
for all $\theta\in\sigma$ and $\Theta$, and $\sigma\in\text{DT}\left(\Theta\right)$.
\end{thm}

\section{Level set representation}

The previous section emphasized support functions as a computational
tool which leads to an MILP formulation for the robust choice function
$\psi\left(X;\,\mathcal{S},\,Y\right)$. In this section, we take
an alternative approach and focus on the so-called level set representation
for quasi-concave choice functions. This result is important because
it generalizes beyond worst-case choice functions and offers a common
framework for representation of multi-attribute choice functions.
Furthermore, this approach has not yet been seen in the work on robust
preference utility/risk models.

\subsection{Risk measures}

Risk measures are the core ingredient for our level set representation.
To begin, we formally define risk measures for the multi-attribute
setting, where we are particularly concerned with convex risk measures
which play a major role in risk-aware optimization (see \cite{artzner1999coherent,Ruszczynski2006a}).
\begin{defn}
\label{def:risk} A function $\mu\mbox{ : }\mathcal{L}\rightarrow\mathbb{R}$
is a convex risk measure if it satisfies

(i) Monotonicity: If $X\leq Y$ then $\mu\left(X\right)\geq\mu\left(Y\right)$.

(ii) Normalization: $\mu(0)=0$.

(iii) Convexity: for any $X,\,Y\in\mathcal{L}$ and $\lambda\in\left[0,\,1\right]$,
$\mu\left(\lambda\,X+\left(1-\lambda\right)Y\right)\leq\lambda\,\mu\left(X\right)+\left(1-\lambda\right)\mu\left(Y\right)$. 
\end{defn}

Since we treat prospects in $\mathcal{L}$ as gains/rewards, our definition
of monotonicity above is the opposite of the typical definition of
monotonicity for losses. Convexity of risk measures is extremely important
for our considerations because of the prominent role of convexity
in optimization. Note that we do not yet stipulate a property of translation
invariance for the multi-attribute setting, we give further commentary
on this issue later.

Any risk measure $\mu\mbox{ : }\mathcal{L}\rightarrow\mathbb{R}$
induces a set of ``acceptable'' prospects in the sense that a prospect
$X$ is acceptable if $\mu\left(X\right)\leq0$, i.e. it has nonpositive
risk according to $\mu$. We formalize this idea in the following
definition. 
\begin{defn}
\label{def:measure_acceptance} Let $\mu\mbox{ : }\mathcal{L}\rightarrow\mathbb{R}$
be a risk measure, the set $\mathcal{A}_{\mu}\triangleq\left\{ X\in\mathcal{L}\mbox{ : }\mu\left(X\right)\leq0\right\} $
is the acceptance set associated with $\mu$. 
\end{defn}

Note that $\mathcal{A}_{\mu}$ is a convex set in $\mathcal{L}$ whenever
$\mu$ is a convex risk measure. Conversely, given an acceptance set
$\mathcal{A}\subset\mathcal{L}$ we can specify a risk measure $\mu_{\mathcal{A},\,d}\left(X\right)\triangleq\inf_{\alpha\in\mathbb{R}}\left\{ X+\alpha\,d\in\mathcal{A}\right\} $
for some $d\in\mathbb{R}^{n}$ with $d>0$. We interpret $\mu_{\mathcal{A},\,d}\left(X\right)d$
as the vector-valued amount that must be added to $X$ to make $X$
acceptable to the decision maker.

For a choice function $\rho\in\mathcal{R}_{iqv}$ which evaluates
the fitness of prospects in $\mathcal{L}$, we seek a family of convex
risk functions $\left\{ \mu_{k}\right\} _{k\in\mathbb{R}}$ such that:
\begin{equation}
\left\{ X\in\mathcal{L}\text{ : }\rho\left(X\right)\geq k\right\} =\left\{ X\in\mathcal{L}\text{ : }\mu_{k}\left(X\right)\leq0\right\} ,\,\forall k\in\mathbb{R}.\label{eq:level}
\end{equation}
Relationship (\ref{eq:level}) means that the upper level sets of
the choice function $\rho$ can be characterized by the acceptance
sets of a family of convex risk functions. The acceptance sets are
closed and convex since $\rho$ is assumed to be upper-semicontinuous
and quasi-concave. From a practical point of view, if a decision maker
with choice function $\rho$ selects satisfaction level $k$, then
the prospect of exceeding $k$ is equivalent to the prospect of the
risk being less than or equal to zero under measure $\mu_{k}$. This
perspective is related to the notion of satisficing measures developed
in \cite{brown2009satisficing}. We interpret (\ref{eq:level}) to
mean that the set of satisfiable prospects can be characterized by
the acceptance sets of a sequence of risk measures.

In the case of relationship (\ref{eq:level}), we have the representation

\begin{equation}
\rho\left(X\right)=\sup\left\{ k\in\mathbb{R}\text{ : }\mu_{k}\left(X\right)\leq0\right\} ,\,\forall X\in\mathcal{L}.\label{eq:representation}
\end{equation}
This equivalence is established in Proposition \ref{prop:Basic_level}
in the Appendix. Representation (\ref{eq:representation}) is valuable
both from a theoretical perspective and a computational one. Theoretically,
it reveals the connection between quasi-concave choice functions and
convex risk measures. Computationally, it allows us to evaluate quasi-concave
choice functions with a sequence of convex feasibility problems and
a bisection algorithm. 
\begin{rem}
A related form of representation (\ref{eq:representation}) is considered
for univariate prospects in \cite{brown2012aspirational}. As an illustrative
example, \cite[Example 5]{brown2012aspirational} considers the case
where the $\left\{ \mu_{k}\right\} _{k\in\mathbb{R}}$ are given by
CVaR. Furthermore, the univariate form of Representation (\ref{eq:representation})
can be viewed as a generalization of the shortfall risk measure from
\cite[Section 4.3]{follmer2002convex}, \cite{acerbi2002coherence},
\cite[Section 3]{weber2006distribution}, and \cite{giesecke2008measuring}.
In this case, for a convex loss function $l\text{ : }\mathbb{R}\rightarrow\mathbb{R}$
we may take $\mu_{k}\left(X\right)=\mathbb{E}\left[l\left(X-k\right)\right]$
for all $k\in\mathbb{R}$. 
\end{rem}

In representation (\ref{eq:representation}), we may also take $\mu_{k}$
to be any support function of $\rho$ at $\rho\left(X\right)=k$.
We now consider the family $\{\mu_{k}\}_{k\in\mathbb{R}}$ more carefully,
and we require it to satisfy the following assumptions. 
\begin{assumption}
\label{assu:representation} $\left\{ \mu_{k}\right\} _{k\in\mathbb{R}}$
is a family of risk measures which satisfies the following: 
\begin{itemize}
\item[(i)] for each fixed $k\in\mathbb{R}$, $\mu_{k}\left(\cdot\right)$ is
monotonically decreasing (non-increasing) on $\mathcal{L}$; 
\item[(ii)] for each fixed $k\in\mathbb{R}$, $\mu_{k}\left(\cdot\right)$ is
convex on $\mathcal{L}$; 
\item[(iii)] for each fixed $k\in\mathbb{R}$, $\mu_{k}(\cdot)$ has closed acceptance
sets; 
\item[(iv)] for each fixed $X\in\mathcal{L}$, $\mu_{k}\left(X\right)$ is monotonically
increasing (non-decreasing) in $k$ over $\mathbb{R}$. 
\end{itemize}
\end{assumption}

Property (iv) means that for fixed $X\in\mathcal{L}$, $\mu_{k_{1}}(X)\leq\mu_{k_{2}}(X)$,
that is, $\mu_{k_{2}}(\cdot)$ assigns a higher risk value $\mu_{k_{1}}(\cdot)$
for $X$.

Our next technical result shows that we can use the set of risk measures
$\{\mu_{k}\}_{k\in\mathbb{R}}$ satisfying Assumption \ref{assu:representation}
to construct a choice function $\rho\in\mathcal{R}_{iqv}$. Specifically,
we may define $\rho\left(X\right)$ to be the highest index level
$k$ such that the corresponding risk of $X$ is acceptable, i.e.
$\mu_{k}\left(X\right)\leq0$. The conclusion follows directly from
Proposition \ref{prop:Basic_choice} in the Appendix. 
\begin{thm}
\label{thm:representation} Suppose Assumption \ref{assu:representation}
holds for $\{\mu_{k}\}_{k\in\mathbb{R}}$ and let 
\begin{equation}
\vartheta\left(X\right)\triangleq\sup\left\{ k\in\mathbb{R}\text{ : }\mu_{k}\left(X\right)\leq0\right\} ,\,\forall X\in\mathcal{L}.\label{eq:representation-1}
\end{equation}
Then $\vartheta$ is upper semi-continuous, increasing, and quasi-concave. 
\end{thm}

In the case when $\mu_{k}(X)\leq0$ for all $k\in\mathbb{R}$, we
have $\rho(X)=+\infty$. In order for $\rho\left(X\right)$ to be
finite valued, $\mu_{k}\left(X\right)$ must increase substantially
w.r.t. increase of $k$. This situation is similar to the case of
utility shortfall risk measures where the underlying loss function
must be strictly increasing from some point, see \cite{follmer2002convex}.

We now consider the reverse implication of the previous theorem, which
shows that there is such a representation (\ref{eq:representation})
for any choice function $\rho\in\mathcal{R}_{iqv}$. As we will see
in the following proof, the choice of this representation is not unique. 
\begin{thm}
\label{thm:representation-1} For any $\rho\in\mathcal{R}_{iqv}$,
there exists a family $\left\{ \mu_{k}\right\} {}_{k\in\mathbb{R}}$
satisfying Assumption \ref{assu:representation} such that 
\[
\rho\left(X\right)=\sup\left\{ k\in\mathbb{R}\text{ : }\mu_{k}\left(X\right)\leq0\right\} ,\,\forall X\in\mathcal{L}.
\]
\end{thm}

\begin{proof}
We provide a constructive proof for the claim and demonstrate that
$\mu_{k}$ can be constructed in three different ways.

(i) Let 
\[
\mu_{k}\left(X\right)=\begin{cases}
0 & \rho\left(X\right)\geq k,\\
\infty & \text{otherwise}.
\end{cases}
\]
Then 
\[
\left\{ X\in\mathcal{L}\text{ : }\rho\left(X\right)\geq k\right\} =\left\{ X\in\mathcal{L}\text{ : }\mu_{k}\left(X\right)\leq0\right\} ,
\]
for all $k\geq0$, and $\rho\left(X\right)=\sup\left\{ k\in\mathbb{R}\text{ : }\mu_{k}\left(X\right)\leq0\right\} $.
It is immediate that the proposed $\left\{ \mu_{k}\right\} _{k\in\mathbb{R}}$
are increasing in $k$, decreasing by monotonicity of $\rho\in\mathcal{R}_{iqv}$,
are convex, and have closed acceptance sets.

(ii) Let 
\[
\mu_{k}\left(X\right)=\text{dist}\left(X,\,\left\{ X\in\mathcal{L}\text{ : }\rho\left(X\right)\geq k\right\} \right).
\]
Then $\left\{ X\in\mathcal{L}\text{ : }\rho\left(X\right)\geq k\right\} =\left\{ X\in\mathcal{L}\text{ : }\mu_{k}\left(X\right)\leq0\right\} $
and $\rho\left(X\right)=\sup\left\{ k\in\mathbb{R}\text{ : }\mu_{k}\left(X\right)\leq0\right\} $.
Again, it is immediate that the proposed $\left\{ \mu_{k}\right\} _{k\in\mathbb{R}}$
satisfy the criteria of Assumption \ref{assu:representation} by choice
of $\rho\in\mathcal{R}_{iqv}$.

(iii) Let $d\in\mathbb{R}^{n}$ with $d>0$ (component-wise strict
inequality), and let 
\[
\mu_{k}\left(X\right)=\inf\left\{ a\in\mathbb{R}\text{ : }\rho\left(X+a\,d\right)\geq k\right\} .
\]
Then $\left\{ X\in\mathcal{L}\text{ : }\rho\left(X\right)\geq k\right\} =\left\{ X\in\mathcal{L}\text{ : }\mu_{k}\left(X\right)\leq0\right\} $
and $\rho\left(X\right)=\sup\left\{ k\in\mathbb{R}\text{ : }\mu_{k}\left(X\right)\leq0\right\} $.
\textcolor{black}{This choice of $d$ has an important interpretation
in multi-attribute decision-making, specifically, this $d$ reveals
the decision maker's endogenous weights among the multiple attributes.
}The desired properties of the proposed $\left\{ \mu_{k}\right\} _{k\in\mathbb{R}}$
then follow from Proposition \ref{prop:Basic_risk}. 
\end{proof}
It follows from the previous two theorems that there is an equivalence
between $\rho\in\mathcal{R}_{iqv}$ and families $\left\{ \mu_{k}\right\} _{k\in\mathbb{R}}$
satisfying Assumption \ref{assu:representation}. Theorems \ref{thm:representation}
and \ref{thm:representation-1} together complete what we call the
``level set representation'' for a choice function $\rho\in\mathcal{R}_{iqv}$.
These theorems are closely related to the results in \cite{brown2009satisficing,brown2012aspirational}
for univariate choice functions, we discuss this relationship further
later in the paper.

\subsection{Connection with support functions}

In this subsection we briefly comment on the connection between our
earlier support function representation and our level set representation. 
\begin{thm}
\label{thm:Connection} Let $\rho\left(X\right)=\inf_{j\in\mathcal{J}}h_{j}\left(X\right)$
where $h_{j}\left(X\right)=h_{j}\left(\vec{X}\right)=\max\left\{ \langle a_{j},\,\vec{X}\rangle+b_{j},\,c_{j}\right\} $
for $a_{j}\in\mathbb{R}^{n\,|\Omega|}$ and $a_{j}\geq0$ for all
$j\in\mathcal{J}$. Assume that there exists $d\in\mathbb{R}^{n}$
such that $\langle a_{j},\,d\rangle>0$ for all $j\in\mathcal{J}$
and let 
\begin{equation}
\mu_{k}\left(X\right)\triangleq\sup_{j\in\mathcal{J}}\inf\left\{ t\in\mathbb{R}\text{ : }h_{j}\left(X+t\,d\right)\geq k\right\} ,\,\forall k\in\mathbb{R}.\label{eq:muk-thm-connection}
\end{equation}
The following assertions hold. 
\begin{itemize}
\item[(i)] $\rho\left(X\right)=\sup\left\{ k\in\mathbb{R}\text{ : }\mu_{k}\left(X\right)\leq0\right\} $
for all $X\in\mathcal{L}$; 
\item[(ii)] $\mu_{k}$ is convex on $\mathcal{L}$ for all $k\in\mathbb{R}$. 
\end{itemize}
\end{thm}

\begin{proof}
Observe first that $h_{j}\left(X+t\,d\right)$ is strictly increasing
in $t$ for each $j\in\mathcal{J}$ under the specified choice of
$d$.

Part (i). By the proof of Theorem \ref{thm:representation-1}, $\rho(X)=\sup\{k\in\mathbb{R}:\tilde{\mu}_{k}(X)\leq0\},$
where 
\[
\tilde{\mu}_{k}(X)\triangleq\inf\left\{ t\in\mathbb{R}\text{ : }\rho\left(X+t\,d\right)\geq k\right\} .
\]
It suffices to show that $\tilde{\mu}_{k}(X)=\mu_{k}(X)$ where the
latter is defined by (\ref{eq:muk-thm-connection}). By definition
\[
\inf\left\{ t\in\mathbb{R}\text{ : }\rho\left(X+t\,d\right)\geq k\right\} =\inf\left\{ t\in\mathbb{R}\text{ : }\inf_{j\in\mathcal{J}}h_{j}\left(X+t\,d\right)\geq k\right\} %=\inf\left\{ a\in\mathbb{R}\text{ : }h_{j}\left(X+t\,d\right)\geqk,\,\forallj\in\mathcal{J}\right\} .
\]
Thus, we are left to show 
\begin{equation}
\inf\left\{ t\in\mathbb{R}\text{ : }\inf_{j\in\mathcal{J}}h_{j}\left(X+t\,d\right)\geq k\right\} =\sup_{j\in\mathcal{J}}\inf\left\{ t\in\mathbb{R}\text{ : }h_{j}\left(X+t\,d\right)\geq k\right\} .\label{eq:exchange-min-max}
\end{equation}
Let $t^{*}$ and $\hat{t}$ denote respectively the optimal value
on the left hand side and right hand side of equation (\ref{eq:exchange-min-max}),
let $t_{j}=\inf\left\{ t\in\mathbb{R}\text{ : }h_{j}\left(X+t\,d\right)\geq k\right\} $.
If $t^{*}=+\infty$, then $\hat{t}=+\infty$. To see this, assume
for the sake of a contradiction that $\hat{t}<+\infty$. Then $t_{j}<+\infty$
for every $j\in\mathcal{J}$. Since $h_{j}(X+td)$ is strictly increasing
in $t$, then for any positive number $\epsilon$, 
\[
h_{j}\left(X+(\hat{t}+\epsilon)\,d\right)\geq h_{j}\left(X+(t_{j}+\epsilon)\,d\right)>h_{j}\left(X+t_{j}\,d\right)\geq k,\forall j\in{\cal J},
\]
which implies $\hat{t}+\epsilon$ is a feasible solution of the minimization
problem at the left hand side of (\ref{eq:exchange-min-max}) and
hence $t^{*}<+\infty$, which gives the desired contradiction.

We now consider the case when $t^{*}<+\infty$. For any $\epsilon>0$,
\[
\inf_{j\in\mathcal{J}}h_{j}\left(X+(t^{*}+\epsilon)\,d\right)\geq k.
\]
Driving $\epsilon$ to $0$, we have that $\inf_{j\in\mathcal{J}}h_{j}\left(X+t^{*}\,d\right)\geq k$
and hence 
\[
h_{j}\left(X+t^{*}\,d\right)\geq k,\forall j\in\mathcal{J}.
\]
This enables us to deduce that $t_{j}\leq t^{*}$ for each $j\in\mathcal{J}$
and subsequently $\hat{t}\leq t^{*}$. To show the reverse inequality,
we note that for any $\epsilon>0$, 
\[
h_{j}\left(X+(\hat{t}+\epsilon)\,d\right)\geq k,\,\forall j\in\mathcal{J},
\]
and hence $\inf_{j\in\mathcal{J}}h_{j}\left(X+\hat{t}\,d\right)\geq k.$
The latter implies $\hat{t}$ is a feasible solution of the minimization
problem on the left hand side of equation (\ref{eq:exchange-min-max})
and hence $t^{*}\leq\hat{t}$. The proof is complete.

Part (ii). We first note that $h_{j}$ is quasi-concave as the maximum
of a linear function and a constant function ($h_{j}$ is automatically
convex, but it is also quasi-concave in this special case). For each
$j\in\mathcal{J}$ and $k\in\mathbb{R}$ we define 
\[
\hat{\mu}_{j,\,k}\left(X\right)\triangleq\inf\left\{ t\in\mathbb{R}\text{ : }h_{j}\left(X+t\,d\right)\geq k\right\} .
\]
Convexity of $\hat{\mu}_{j,\,k}$ follows by Proposition \ref{prop:Basic_risk}.
The desired conclusion follows by noting that $\mu_{k}$ itself is
the supremum of the convex functions $\left\{ \hat{\mu}_{j,\,k}\right\} _{j\in\mathcal{J}}$,
which preserves convexity. 
\end{proof}
Theorem \ref{thm:Connection} shows that we may construct the level
set representation (\ref{eq:representation}) for $\rho\in\mathcal{R}_{iqv}$
by using the support functions of $\rho$. However, from computational
point view, it might be difficult to identify the optimal (lowest)
support function which majorizes $\rho(X)$ when $\rho$ has a general
structure. This motivates us to consider a lower approximation of
$\mu_{k}(X)$ constructed by using a single piece of a support function,
which is not necessarily optimal, but that can be obtained by calculating
an upper sub-gradient of $\rho$. Let us denote this function by $h_{j}(X)=\max\left(\langle a_{j},\,\vec{X}\rangle+b_{j},\,c_{j}\right)$
and define 
\[
\nu_{k}(X)\triangleq\inf\left\{ t\in\mathbb{R}\text{ : }h_{j}\left(X+t\,d\right)\geq k\right\} .
\]
Obviously $\nu_{k}(X)\leq\mu_{k}(X)$ and through (\ref{eq:level}),
we have 
\begin{equation}
\left\{ X\in\mathcal{L}\text{ : }\rho\left(X\right)\geq k\right\} =\left\{ X\in\mathcal{L}\text{ : }\mu_{k}\left(X\right)\leq0\right\} \subset\left\{ X\in\mathcal{L}\text{ : }\nu_{k}\left(X\right)\leq0\right\} ,\,\forall k\in\mathbb{R},\label{eq:level-1}
\end{equation}
which implies that the upper level set of $\rho$ at level $k$ is
contained by the lower level set of $\nu_{k}$ at level $0$. By developing
a sequence of functions $\left\{ \nu_{k}\right\} _{k\in\mathbb{R}}$
appropriately, we will be able to use the intersection of the lower
level sets of $\nu_{k}$, $k=1,\cdots,$ to approximate the maximizer
of $\rho$ (the upper level set of $\rho$ at its maximum). This procedure
is closely related to the idea of the so-called level function method
in \cite{Xu2001} which we will discuss in detail in Section 6 for
solving Problem (\ref{OPTIMIZATION}). Note also that in the case
when $k=c$, we can obtain an explicit form of $\nu_{k}\left(X\right)$
with $\nu_{c}\left(X\right)=\left(\langle a_{j},\,X\rangle+b_{j}-c_{j}\right)/\langle a_{j},\,d\rangle$.
So, if we set $\langle a_{j},d\rangle=1$, then we obtain $\nu_{c}\left(X\right)=\langle a_{j},\,X\rangle+b_{j}-c_{j}$.
Moreover, by setting $b_{j}=c_{j}-\langle a_{j},\,X_{c}\rangle$,
we have $\nu_{c}\left(X\right)=\langle a_{j},\,X-X_{c}\rangle$ which
is a level function (see Definition \ref{def:level-function}).

\subsection{The role of targets}

In this section we explore the role of \textit{targets} in expressing
decision maker preferences. This issue is originally investigated
in \cite{brown2012aspirational}, and here we extend it to the multi-attribute
setting. A target is a benchmark (typically a monetary amount or level
of reward) that the decision maker wishes to achieve. We can assess
a prospect $X$ in terms of its ability to meet a target. Target-based
decision making is further elaborated upon in \cite{lanzillotti1958pricing,mao1970survey,fishburn1977mean}
and forms the basis for the satisficing and aspirational preferences
in \cite{brown2009satisficing,brown2012aspirational}. We remark that
this notion of target is different from the benchmarks that appear
in the literature on stochastic dominance constraints. Here, we have
a continuum of scalar targets for all satiation levels, while in the
case of stochastic dominance constraints the benchmark is a single
fixed prospect.

Next we formally define a target and related concepts for prospects
$X\in\mathcal{L}$. 
\begin{defn}
\label{def:target} (i) A target $\tau\in\mathbb{R}^{n}$ is a desired
level of gain/reward.

(ii) For a prospect $X$, $X-\tau$ is the \textit{target premium}.

(iii) For a risk measure $\mu\mbox{ : }\mathcal{L}\rightarrow\mathbb{R}$,
$\mu\left(X-\tau\right)$ is the risk associated with the target premium
$X-\tau$. 
\end{defn}

The next two theorems parallel Theorems \ref{thm:representation}
and \ref{thm:representation-1}, except that now targets appear explicitly.
The following result shows that given risk measures and targets, we
may construct a choice function in $\mathcal{R}_{iqv}$. The proof
of the following theorem is based on Proposition \ref{prop:Basic_choice}
in the Appendix. In the following theorem and throughout this section,
we let $\left\{ \tau\left(k\right)\right\} _{k\in\mathbb{R}}\subset\mathbb{R}^{n}$
denote a family of targets indexed by $k\in\mathbb{R}$. In this sense
we can view the target mapping $\tau\left(k\right)$ as a function
of $k$, i.e. $\tau\text{ : }\mathbb{R}\rightarrow\mathbb{R}^{n}$. 
\begin{thm}
\label{thm:target} Suppose we are given a collection $\left\{ \mu_{k}\right\} _{k\in\mathbb{R}}$
of convex risk measures with closed acceptance sets and targets $\left\{ \tau\left(k\right)\right\} _{k\in\mathbb{R}}\subset\mathbb{R}^{n}$.
We define $\tilde{\mu}_{k}\left(X\right)\triangleq\mu_{k}\left(X-\tau\left(k\right)\right)$
and suppose:

(i) all $\left\{ \tilde{\mu}_{k}\right\} _{k\in\mathbb{R}}$ have
closed acceptance sets;

(ii) for any $X\in\mathcal{L}$, $\tilde{\mu}_{k}\left(X\right)$
as a function of $k$ is non-decreasing and left continuous;

(iii) for any $X\in\mathcal{L}$, there is $k\in\mathbb{R}$ such
that $\tilde{\mu}_{k}\left(X\right)\leq0$.\\
 Then 
\begin{equation}
\psi\left(X\right)\triangleq\sup\left\{ k\in\mathbb{R}\mbox{ : }\mu_{k}\left(X-\tau\left(k\right)\right)\leq0\right\} ,\,\forall X\in\mathcal{L},\label{eq:target}
\end{equation}
is upper semi-continuous, monotonic, and quasi-concave. 
\end{thm}

The next result is the reverse implication of Theorem \ref{thm:target},
it shows how to construct risk measures and targets so that we may
recover any choice function $\rho\in\mathcal{R}_{iqv}$ in the form
of representation (\ref{eq:target}). Now, we allow the targets to
be any measurable selection from the level sets of $\rho$. In the
following, let $\text{bd}\left(\cdot\right)$ denote the boundary
of a set in $\mathcal{L}$ with respect to the supremum norm topology. 
\begin{thm}
\label{thm:target-1} Let $\rho\in\mathcal{R}_{iqv}$ and choose 
\[
\tau\left(k\right)\in\text{bd}\left\{ X\in\mathcal{L}\mbox{ : }\rho\left(X\right)\geq k\right\} ,\,\forall k\in\mathbb{R}.
\]
Let
\[
\mu_{k}\left(X\right)\triangleq\inf\left\{ \alpha\in\mathbb{R}\mbox{ : }\rho\left(X+\alpha\,\tau\left(k\right)\right)\geq k\right\} -1,\,\forall X\in\mathcal{L}.
\]
Then the following assertions hold.

(i) $\left\{ \mu_{k}\right\} _{k\in\mathbb{R}}$ are monotonic, translation
invariant (along $\tau\left(k\right)$), normalized, convex, and have
closed acceptance sets.

(ii) $\rho\left(X\right)=\sup\left\{ k\in\mathbb{R}\mbox{ : }\mu_{k}\left(X-\tau\left(k\right)\right)\leq0\right\} $
for all $X\in\mathcal{L}$. 
\end{thm}

\begin{proof}
(i) The desired properties again follow from Proposition \ref{prop:Basic_risk}
with some minor variations.

\textit{Translation invariance:} For any $c\in\mathbb{R}$, we have
\begin{align*}
\mu_{k}\left(X+c\,\tau\left(k\right)\right)=\, & \inf\left\{ \alpha\in\mathbb{R}\mbox{ : }\rho\left(X+c\,\tau\left(k\right)+\alpha\,\tau\left(k\right)\right)\geq k\right\} -1\\
=\, & \inf\left\{ \alpha-c\mbox{ : }\rho\left(X+\alpha\,\tau\left(k\right)\right)\geq k\right\} -1\\
=\, & \inf\left\{ \alpha\in\mathbb{R}\mbox{ : }\rho\left(X+\alpha\,\tau\left(k\right)\right)\geq k\right\} -1-c\\
=\, & \mu_{k}\left(X\right)-c.
\end{align*}

\textit{Normalization:} Normalization follows immediately by definition
since: 
\[
\mu_{k}\left(0\right)=\inf\left\{ \alpha\in\mathbb{R}\mbox{ : }\rho\left(\alpha\,\tau\left(k\right)\right)\geq k\right\} -1=0,
\]
since $\inf\left\{ \alpha\in\mathbb{R}\mbox{ : }\rho\left(\alpha\,\tau\left(k\right)\right)\geq k\right\} =1$
by choice of $\tau\left(k\right)$ in $\left\{ d\in\mathbb{R}^{n}\mbox{ : }\rho\left(d\right)=k\right\} $.

(ii) We first verify that 
\begin{equation}
\left\{ X\in\mathcal{L}\text{ : }\rho\left(X\right)\geq k\right\} =\left\{ X\in\mathcal{L}\text{ : }\mu_{k}\left(X-\tau\left(k\right)\right)\leq0\right\} ,\,\forall k\in\mathbb{R}.\label{eq:target_equivalence}
\end{equation}
If $\rho\left(X\right)\geq k$ then 
\[
\mu_{k}\left(X-\tau\left(k\right)\right)=\inf\left\{ \alpha\in\mathbb{R}\mbox{ : }\rho\left(X+\alpha\,\tau\left(k\right)\right)\geq k\right\} +1-1\leq0,
\]
where the equality follows by translation invariance of $\mu_{k}$
along $\tau\left(k\right)$, and the inequality follows since $\inf\left\{ \alpha\in\mathbb{R}\mbox{ : }\rho\left(X+\alpha\,\tau\left(k\right)\right)\geq k\right\} \leq0$
when $\rho\left(X\right)\geq k$. Conversely, suppose $\mu_{k}\left(X-\tau\left(k\right)\right)\leq0$,
then we must have $\inf\left\{ \alpha\in\mathbb{R}\mbox{ : }\rho\left(X+\alpha\,\tau\left(k\right)\right)\geq k\right\} \leq0$
which implies $\rho\left(X\right)\geq k$. Now we show that the desired
conclusion follows from equivalence (\ref{eq:target_equivalence})
by Proposition \ref{prop:Basic_level} in the Appendix. 
\end{proof}
We now develop some concrete examples of our main development in this
section. As our first example, we extend some univariate expected
utility choice functions to the multi-attribute setting. 
\begin{example}
We first take utility functions $u_{i}\mbox{ : }\mathbb{R}\rightarrow\mathbb{R}$
which are continuous, monotonically nondecreasing, concave, and satisfy
$u_{i}\left(0\right)=0$ for $i=1,\ldots,\,n$. For weights $\kappa_{1},\ldots,\,\kappa_{n}\geq0$,
we may then consider the (additive expected utility) choice function
$\rho\left(X\right)=\sum_{i=1}^{n}\kappa_{i}\mathbb{E}\left[u_{i}\left(X_{i}\right)\right]$
on $\mathcal{L}$ where the expectation $\mathbb{E}\left[\cdot\right]$
is taken with respect to some fixed distribution on $\Omega$. We
may then leverage \cite[Example 1]{brown2012aspirational} for each
component $\mathbb{E}\left[u_{i}\left(X_{i}\right)\right]$ for $i=1,\ldots,\,n$.

As a generalization, we now take a multi-attribute utility function
$u\mbox{ : }\mathbb{R}^{n}\rightarrow\mathbb{R}$ that is continuous,
monotonically nondecreasing, concave, and satisfies $u\left(0\right)=0$.
The resulting choice function is the subjective expected utility $\rho\left(X\right)=\mathbb{E}\left[u\left(X\right)\right]$,
which gives either way of constructing corresponding risk functions.
We have 
\[
\tau(k)=\inf\left\{ a\in\mathbb{R}\text{ : }u(a\,d)\geq k\right\} ,
\]
and 
\begin{align*}
\mu_{k}\left(X\right)=\, & \inf\left\{ a\in\mathbb{R}\mbox{ : }\mathbb{E}\left[u\left(X+a\,d\right)\right]\geq k\right\} -\tau(k)\\
=\, & -\left(\sup\left\{ a\in\mathbb{R}:\,\mathbb{E}\left[u\left(X-a\,d\right)\right]\geq k\right\} +\tau(k)\right).
\end{align*}
\end{example}

We may apply this same reasoning to some classical risk measures as
follows. 
\begin{example}
The Optimized Certainty Equivalent (OCE) (which includes conditional
value-at-risk as a special case) is defined in \cite{Ben-Tal2007}
as follows

\[
S_{u}(X)\triangleq\sup_{\eta\in\mathbb{R}}\left\{ \eta+\mathbb{E}\left[u\left(X-\eta\right)\right]\right\} ,
\]
where $u$ is any continuous, monotonically nondecreasing, and concave
utility function that satisfies $u\left(0\right)=0$. OCE has the
following interpretation: suppose a decision maker faces a future
uncertain income of $X$ dollars, and can consume part of $X$ at
present. If he chooses to consume $\eta$ dollars, then the resulting
present value of $X$ is $\eta+\mathbb{E}\left[u\left(X-\eta\right)\right]$.
Thus, the sure (present) value of $X$ is $S_{u}(X)$.

To extend OCE to multi-attribute case, we now take a multi-attribute
utility function $u\text{ : }\mathbb{R}^{n}\rightarrow\mathbb{R}$
satisfying the same conditions as for the univariate case. We may
then define the targets 
\[
\tau\left(k\right)=\inf\left\{ \alpha\in\mathbb{R}\text{ : }\sup_{\eta\in\mathbb{R}}\left\{ \eta+u\left(\alpha\,d-\eta\right)\right\} \geq k\right\} ,
\]
and risk measures

\[
\mu_{k}\left(X\right)=\inf\left\{ \alpha\in\mathbb{R}\text{ : }\sup_{\eta\in\mathbb{R}}\left\{ \eta+\mathbb{E}\left[u\left(X+\alpha\,d-\eta\right)\right]\right\} \geq k\right\} -\tau\left(k\right)
\]
to use in our representation. 
\end{example}

Some risk measures, such as ratio type risk measures, naturally lead
to a representation of the form (\ref{eq:representation}). 
\begin{example}
A robust variant of reward-risk ratio measures is developed in \cite{liu2017distributionally}.
The classical reward-risk ratio measure (on $\mathcal{L}$ for $n=1$)
which penalizes downward variations is 
\[
\rho\left(X\right)=\frac{\mathbb{E}\left[X-Y\right]}{\mathbb{E}\left[\left(Y-X\right)_{+}\right]},
\]
where $Y$ is a benchmark return. The epigraphical formulation of
this risk measure is
\[
\rho\left(X\right)=\sup_{\tau\in\mathbb{R}}\left\{ \tau\text{ : }\mathbb{E}\left[X-Y-\tau\left(Y-X\right)_{+}\right]\geq0\right\} ,
\]
which is quite similar to our representation (\ref{eq:representation})
with $\mu_{\tau}\left(X\right)=\mathbb{E}\left[X-Y-\tau\left(Y-X\right)_{+}\right]$.
In \cite{liu2017distributionally}, the authors focus on distributional
ambiguity, while we can adapt their technique to deal with multi-attribute
prospects on $\mathcal{L}$ for $n\geq2$. In particular, for any
scalarization function $\varphi\text{ : }\mathbb{R}^{n}\times\mathbb{R}^{n}\rightarrow\mathbb{R}$
(such as $\varphi\left(w,\,X\right)=\sum_{i=1}^{n}w_{i}X_{i}$) we
may consider the choice function 
\[
\rho\left(X\right)=\inf_{w\in\mathcal{W}}\frac{\mathbb{E}\left[\varphi\left(w,\,X\right)-\varphi\left(w,\,Y\right)\right]}{\mathbb{E}\left[\left(\varphi\left(w,\,Y\right)-\varphi\left(w,\,X\right)\right)_{+}\right]},
\]
where $\mathcal{W}\subset\mathbb{R}^{n}$ is a set of weights. Then,
using a similar analysis as in \cite[Proposition 2.1]{liu2017distributionally},
we can obtain 
\begin{align*}
\rho\left(X\right)=\sup_{\tau\in\mathbb{R}}\, & \tau\\
\text{s.t.}\, & \inf_{w\in\mathcal{W}}\mathbb{E}\left[\varphi\left(w,\,X\right)-\varphi\left(w,\,Y\right)-\tau\left(\varphi\left(w,\,Y\right)-\varphi\left(w,\,X\right)\right)_{+}\right]\geq0.
\end{align*}
In this case, the family of risk measures in representation (\ref{eq:representation})
is 
\[
\mu_{\tau}\left(X\right)=\inf_{w\in\mathcal{W}}\mathbb{E}\left[\varphi\left(w,\,X\right)-\varphi\left(w,\,Y\right)-\tau\left(\varphi\left(w,\,Y\right)-\varphi\left(w,\,X\right)\right)_{+}\right],
\]
which includes the inner minimization over weights. 
\end{example}

In \cite{brown2009satisficing,brown2012aspirational}, the authors
use a specific construction of the targets. In our case, we require
a ``distinguished'' strictly positive (component-wise) direction
$d\in\mathbb{R}^{n}$ along which we may increase or decrease prospects
in $\mathcal{L}$. No such direction was needed because this work
lies in the univariate setting. This distinguished direction leads
to a notion of translation invariance for multivariate risk measures.
We use the notation $\left\{ \upsilon\left(k\right)\right\} _{k\in\mathbb{R}}\subset\mathbb{R}$
to emphasize that we are only referring to aspiration levels along
this distinguished direction $d$.
\begin{thm}
Let $d\in\mathbb{R}^{n}$ with $d>0$ (component-wise). Given $\rho\in\mathcal{R}_{iqv}$,
define

\[
\upsilon\left(k\right)\triangleq\inf\left\{ \alpha\in\mathbb{R}\mbox{ : }\rho\left(\alpha\,d\right)\geq k\right\} ,
\]
and $\mu_{k}\text{ : }\mathcal{L}\rightarrow\mathbb{R}$ via 
\[
\mu_{k}\left(X\right)\triangleq\inf\left\{ \alpha\in\mathbb{R}\mbox{ : }\rho\left(X+\alpha\,d\right)\geq k\right\} -\upsilon\left(k\right),\,\forall X\in\mathcal{L}.
\]
Then we have:

(i) $\left\{ \mu_{k}\right\} _{k\in\mathbb{R}}$ are monotonic, translation
invariant (along $d$), normalized, convex, and have closed acceptance
sets.

(ii) $\rho\left(X\right)=\sup\left\{ k\in\mathbb{R}\mbox{ : }\mu_{k}\left(X-\upsilon\left(k\right)d\right)\leq0\right\} $
for all $X\in\mathcal{L}$. 
\end{thm}

\begin{proof}
Part (i). All of the desired properties of $\left\{ \mu_{k}\right\} _{k\in\mathbb{R}}$
follow from Proposition \ref{prop:Basic_risk} except for translation
invariance, which is new. For any $c\in\mathbb{R}$, we have 
\begin{align*}
\mu_{k}\left(X+c\,d\right)=\, & \inf\left\{ \alpha\in\mathbb{R}\mbox{ : }\rho\left(X+c\,d+\alpha\,d\right)\leq k\right\} -\upsilon\left(k\right)\\
=\, & \inf\left\{ \alpha-c\mbox{ : }\rho\left(X+\alpha\,d\right)\leq k\right\} -\upsilon\left(k\right)\\
=\, & \inf\left\{ \alpha\in\mathbb{R}\mbox{ : }\rho\left(X+\alpha\,d\right)\leq k\right\} -c-\upsilon\left(k\right)\\
=\, & \mu_{k}\left(X\right)-c.
\end{align*}

Part (ii). Now we verify that $\left\{ X\in\mathcal{L}\text{ : }\rho\left(X\right)\geq k\right\} \Leftrightarrow\left\{ X\in\mathcal{L}\text{ : }\mu_{k}\left(X-\upsilon\left(k\right)d\right)\leq0\right\} $.
If $\rho\left(X\right)\geq k$ then 
\begin{align*}
\mu_{k}\left(X-\upsilon\left(k\right)d\right)=\, & \mu_{k}\left(X\right)+\upsilon\left(k\right)\\
=\, & \inf\left\{ \alpha\in\mathbb{R}\mbox{ : }\rho\left(X+\alpha\,d\right)\geq k\right\} -\upsilon\left(k\right)+\upsilon\left(k\right)\\
\leq\, & 0,
\end{align*}
since $\inf\left\{ \alpha\in\mathbb{R}\mbox{ : }\rho\left(X+\alpha\,d\right)\geq k\right\} \leq0$
when $\rho\left(X\right)\geq k$. Conversely, suppose $\mu_{k}\left(X-\upsilon\left(k\right)d\right)\leq0$,
then we must have $\inf\left\{ \alpha\in\mathbb{R}\mbox{ : }\rho\left(X+\alpha\,d\right)\geq k\right\} \leq0$
which implies $\rho\left(X\right)\geq k$. 
\end{proof}

\section{Optimization of quasi-concave choice functions}

With the preceding discussions on the support function representation
of the robust choice function in Section 4 and the level set representation
in Section 5, we are now ready to develop a numerical procedure for
the maximin robust preference optimization Problem (\ref{OPTIMIZATION}),
\begin{equation}
\max_{z\in\mathcal{Z}}\psi\left(G\left(z\right);\,\mathcal{S},\,Y\right)=\max_{z\in\mathcal{Z}}\inf_{\rho\in\mathcal{S}}\left\{ \rho\left(G\left(z\right)\right)-\rho\left(Y\right)\right\} .\label{eq:OPTIMIZATION-1}
\end{equation}
Observe that if $G$ is a linear function of $z$ in every scenario,
then we will be able to solve (\ref{OPTIMIZATION}) by solving a mixed-integer
linear programming problem as we outlined in Section 4. Here we concentrate
on the case where $G$ is convex in $z$ component-wise almost surely.
It is easy to verify that $\psi\left(G\left(z\right);\,\mathcal{S},\,Y\right)$
is then quasi-concave in $z$. We will combine the support function
approach and level set representation approach and put them in the
framework of the so-called level function method \cite{Xu2001}. We
start with a formal definition of level functions of a quasi-concave
function.
\begin{defn}[Level function]
\label{def:level-function} Let $f\text{ : }\mathbb{R}^{n}\rightarrow\mathbb{R}$
be a continuous real-valued quasi-concave function. A function $\sigma\text{ : }\mathcal{L}\rightarrow\mathbb{R}$
is called a level function of $\rho$ at $x$ if it satisfies:

(i) $\sigma(x)=0$;

(ii) $\sigma$ is a continuous concave function;

(iii) $T_{f}(f(x))\subset T_{\sigma}(0)$, where $T_{f}\left(\alpha\right)=\{x\in\mathbb{R}^{d}\text{ : }f(x)>\alpha\}$
denotes the strict upper level set of $f$ at level $f(x)$. 
\end{defn}

The basic idea of the level function method can be described as follows.
At starting point $x_{0}$, we calculate a level function $\sigma_{0}$
at the point. This can be achieved by calculating an upper sub-gradient
$g_{0}\in\partial^{+}f(x_{0})$ and define $\sigma_{0}(x)\triangleq\langle g_{0},x-x_{0}\rangle$.
We can then minimize $\sigma_{0}(x)$ over the feasible set of the
associated maximization problem (which of course must be compact and
convex) and denote the minimizer as the next iterate $x_{1}$. Next,
we calculate a level function of $f$ at $x_{1}$, minimize $\max\left\{ \sigma_{0}(x),\,\sigma_{1}(x)\right\} $
and use the minimizer as the next iterate. In this process, we will
have to evaluate $f$ at each iterate in order to calculate a sub-gradient
of $f$ and construct a level function there.

To see how this process can be applied to solve (\ref{eq:OPTIMIZATION-1}),
we need to discuss how to construct a level function at each iterate.
Let $(a_{i},b_{i},c_{i})$ denote an optimal solution of Problem (\ref{ROBUST})
- (\ref{ROBUST-6}) at $z_{i}$. Then the function $\text{\ensuremath{\max}}\left\{ \left\langle a_{i},\,G(z)\right\rangle +b_{i},\,c_{i}\right\} $
majorizes $\psi\left(G\left(z\right);\,\mathcal{S},\,Y\right)$ and
they coincide at $z=z_{i}$, that is, $c_{i}=\psi\left(G\left(z_{i}\right);\,\mathcal{S},\,Y\right)$.
Let $\sigma_{z_{i}}(z)\triangleq\left\langle a_{i},\,G(z)-G(z_{i})\right\rangle $,
then $\sigma_{z_{i}}(z_{i})=0$ and it is a convex function. Moreover,
since 
\[
\text{\ensuremath{\max}}\left(\left\langle a_{i},\,G(z)\right\rangle +b_{i},\,c_{i}\right)\geq\text{\ensuremath{\max}}\left(\left\langle a_{i},\,G(z_{i})\right\rangle +b_{i},\,c_{i}\right)=\psi\left(G\left(z_{i}\right);\,\mathcal{S},\,Y\right)
\]
for all $z\in Z$, we have that
\[
\sigma_{z_{i}}(z)\triangleq\left\langle a_{i},\,G(z)-G(z_{i})\right\rangle \geq0,\forall z\in T_{\psi\left(G\left(\cdot\right);\,\mathcal{S},\,Y\right)}(c_{i}).
\]
This computation shows that $\sigma_{z_{i}}$ is a level function
of $\psi\left(G\left(\cdot\right);\,\mathcal{S},\,Y\right)$ at $z_{i}$.
Note that the way level functions are used here is closely linked
to the level set representation that we discussion in Section 5.

We are now ready to present an algorithm for solving (\ref{eq:OPTIMIZATION-1})
based on the projected level function method \cite[Algorithm 3.3]{Xu2001}.

\begin{algorithm}
\textbf{Step 1}: Select a starting point $z_{0}\in\mathcal{Z}$, set
$i=0$; specify tolerance level $\epsilon$;

\textbf{Step 2}: Calculate a level function $\sigma_{z_{i}}(z)$ of
$\psi\left(G\left(\cdot\right);\,\mathcal{S},\,Y\right)$, and set
$\sigma_{i}(z):=\inf\left\{ \sigma_{i-1}(z),\,\sigma_{z_{i}}(z)\right\} ,$
where $\sigma_{-1}(z):=\infty$. Let 
\begin{equation}
z_{i}=\arg\max_{z\in\mathcal{Z}}\left\{ \psi\left(G(z_{j});\,\mathcal{S},\,Y\right)\text{ : }j=1,\,...,\,i\right\} ,\label{discretize}
\end{equation}

and $z_{i+1}\in\Pi_{Q_{i}}\left[z_{i}\right]$, where $Q_{i}:=\left\{ z\in\mathcal{Z}\text{ : }\sigma_{i}(z)\geq\lambda\Delta(i)\right\} $, $\Delta(i)=\max_{z\in\mathcal{Z}}\sigma_{i}(z)$, and $\Pi_{Q}\left[\cdot\right]$
denotes the projection operator onto a set $Q$.

\textbf{Step 3}: If $\Delta(i)\leq\epsilon$, stop; otherwise, set
$i:=i+1$, go to Step 2.

\caption{\label{alg:Level}Projected Level Function Method}
\end{algorithm}

In Algorithm \ref{alg:Level}, when $\lambda=1$, $Q_{i}$ is the
set of minimizers of $\sigma_{i}$ over $\mathcal{Z}$. Generally,
we assume that the constant $\lambda$ belongs to $(0,\,1)$. The
worst-case choice function at point $z_{i}\in\mathcal{Z}$ i.e. the
minimizer in $\psi(G\left(z_{i}\right);\,\mathcal{S},\,Y)$, is modeled
given by the solution of Problem (\ref{ROBUST}) - (\ref{ROBUST-6})
at $z_{i}$.

Based on \cite[Proposition 3.1]{Xu2001}, if Algorithm \ref{alg:Level}
terminates at iteration $i$, then the global maximum of $\psi\left((G\left(\cdot\right);\,\mathcal{S},\,Y\right)$
has been obtained. Let $\{z_{i}\}_{i\geq0}$ be the sequence generated
by Algorithm 1. By \cite[Proposition 3.4]{Xu2001}, if $\lim_{i\rightarrow\infty}\Delta(i)=0$,
then there exists a subsequence of $\{z_{i}\}_{i\geq0}$ converging
to a global maximizer of $\psi\left((G\left(\cdot\right);\,\mathcal{S},\,Y\right)$
over $\mathcal{Z}$. Further, we have the convergence rate of Algorithm
\ref{alg:Level}, presented next. 
\begin{thm}
\label{Convergence rate} \cite[Theorem 3.3]{Xu2001} Let $\{z_{i}\}_{i\geq0}$
be a sequence generated by Algorithm \ref{alg:Level}. Assume that
the sequence of level functions $\{\sigma_{z_{i}}(z)\}_{i\geq0}$
is uniformly Lipschitz on $\mathcal{Z}$ with constant $K$. Then
for any $\epsilon>0$, $\Delta(i)\leq\epsilon$ for $i>K^{2}D^{2}\epsilon^{-2}\lambda^{-2}(1-\lambda^{2})^{-1}$
where $D$ is the diameter of convex set $\mathcal{Z}$.
\end{thm}

\section{Application to homeland security}

In \cite{Phillips2007}, the authors highlight the need for multiple
criteria when investing limited resources. In this section, we apply
our methods to a real world multi-attribute budget allocation problem
in homeland security where we are investing resources to protect against
terrorist threats. This application is based on the case study originally
investigated in \cite{Hu2011}. Specifically, we want to allocate
resources among ten major cities in the United States under the Urban
Areas Security Initiative (UASI) of the Department of Homeland Security
(DHS). DHS is the principle decision maker in this application.

\subsection{Problem description}

We begin by introducing the problem setup and notations. We have $m=10$
cities among which we need to allocate the budget. There are three
possible underlying loss scenarios corresponding to different levels
of terrorist attacks: 
\[
\left\{ \text{reduced loss, standard loss, increased loss}\right\} .
\]
We remark on the probability distribution over these scenarios later
in this section. There are $n=4$ attributes for measuring the effect
of a terrorist attack:
\begin{itemize}
\item \textit{property loss} measures the impact on economic structures
and individual property;
\item \textit{fatalities} measures the loss of human lives;
\item \textit{air departures} measures the number of outgoing flights from
the city's major airport;
\item \textit{bridge traffic} measures vehicle movement on the city's major
bridges.
\end{itemize}
These four attributes also appear in the case study in \cite{Hu2011}
and give a comprehensive measure of the health of a city.

We denote the budget allocated to city $j$ corresponding to attribute
$i$ by $z_{ij}$, and allocations $z=\left(z_{ij}\right)$ must satisfy
$z\in\mathcal{Z}\subseteq\mathbb{R}^{(n\times m)}$, where $\mathcal{Z}$
represents the set of feasible resource allocations:
\[
\mathcal{Z}\triangleq\left\{ z\in\mathbb{R}^{(n\times m)}\text{ : }\sum_{i=1}^{n}\sum_{j=1}^{m}z_{ij}\leq B\right\} ,
\]
and $B$ is the overall budget. Based on the budget allocation models
in \cite{Nikoofal2012,MacKenzie2012}, we may use an exponential function
$g\text{ : }\mathbb{R}\rightarrow\mathbb{R}$ to measure the effectiveness
of investment. Let $v_{ij}$ denote the target for city $j$ and attribute
$i$ (for instance, this target can be chosen to be the possible maximum
value of loss for this attribute - indeed, we take the target to be
this maximum value in our upcoming experiments). Let $\delta\in(0,\,1]$
denote the effectiveness ratio of investment to obtain the function
$g_{ij}$ for each $z_{ij}$ defined as
\begin{equation}
g_{ij}\left(z_{ij}\right)\triangleq v_{ij}\left(1-\exp\left(-\delta\,z_{ij}\right)\right),\label{effectiveness}
\end{equation}
for all attributes $i$ and cities $j$.

The random loss at city $j$ with respect to attribute $i$ is denoted
by $C_{ij}$. We have an $(n\times m)$ random loss matrix $C$ that
captures the loss for all attributes in all cities. Given a particular
budget allocation $z\in\mathcal{Z}$, the shortfall with respect to
attribute $i$ is measured by the function 
\[
C_{i}(z)\triangleq\sum_{j=1}^{m}(C_{ij}-g(z_{ij}))_{+},
\]
where $(\cdot)_{+}\triangleq\max\left\{ \cdot,\,0\right\} $. The
quantity $C_{i}\left(z\right)$ can be viewed as the sum of the shortfall
in attribute $i$ over all cities. If $g(z_{ij})\geq C_{ij}$, then
the allocated resources exceed each cities' need and there is no shortfall.
We combine the shortfall for each attribute into the vector-valued
mapping 
\[
C\left(z\right)=\left(C_{1}\left(z\right),\ldots,\,C_{n}\left(z\right)\right).
\]
The negative value of the shortfall, i.e. $-C\left(z\right)$ can
be viewed as the ``reward''. It is clear that $-C(z)$ is component-wise
concave in $z$ (since each component $-C_{i}\left(z\right)$ is a
concave function of $z$).

We consider two models in this case study. The first is risk-neutral
and minimizes expected shortfall:

\begin{equation}
\min_{z\in\mathcal{Z}}\mathbb{E}\left[\sum_{i=1}^{n}\sum_{j=1}^{m}(C_{ij}-g(z_{ij}))_{+}\right].\label{NEUTRAL}
\end{equation}
The second is our robust preference model with the special ambiguity
set $\mathcal{S}$:
\begin{equation}
\max_{z\in\mathcal{Z}}\inf_{\rho\in\mathcal{S}}\rho(-C(z)).\label{DHS}
\end{equation}
Problem (\ref{DHS}) is especially relevant at the governmental level
where the preferences of many different stakeholders must be accommodated.
In \cite{raey}, the stakeholders are noted to include: federal agencies;
department and component officials; state, local, and tribal governments;
the private sector; academics; and policy experts.

\subsection{Data}

We now describe the data set used in our experiments. First we describe
the data for our four attributes: (i) property loss; (ii) fatalities;
(iii) air departures; and (iv) bridge traffic. Table 1 represents
the data obtained from \cite{Willis} and \cite{Bier2008} (these
data were previously applied in \cite{Hu2011}). The ten cites in
Table 1 received 40\% of the total UASI budget in 2004 and 60\% of
the total UASI budget in 2009. The number of air departures in these
ten urban areas accounts for roughly one third of the total air departures
in the United States (see \cite{Bier2008}). In Table 1, the data
on property loss and fatalities are recorded as random losses in millions
of dollars. The data on air departures and bridge traffic are recorded
as daily averages. 

\begin{table}
\begin{centering}
{\scriptsize{}}%
\begin{tabular}{lcccccccc}
\hline 
 & \multicolumn{3}{c}{\emph{\scriptsize{}Property losses (\$ million)}} & \multicolumn{3}{c}{\emph{\scriptsize{}Fatalities (\$ million)}} & \multicolumn{1}{c}{\emph{\scriptsize{}Average daily air }} & \multicolumn{1}{c}{\emph{\scriptsize{}Average daily}}\tabularnewline
\cline{2-7} 
\emph{\scriptsize{}Urban area} & \multirow{1}{*}{\emph{\scriptsize{}Standard}} & \emph{\scriptsize{}Reduced} & \emph{\scriptsize{}Increased} & \emph{\scriptsize{}Standard} & \emph{\scriptsize{}Reduced} & \emph{\scriptsize{}Increased} & \emph{\scriptsize{}departures} & \emph{\scriptsize{} bridge traffic}\tabularnewline
\hline 
{\scriptsize{}New York} & {\scriptsize{}413} & {\scriptsize{}265} & {\scriptsize{}550} & {\scriptsize{}304} & {\scriptsize{}221} & {\scriptsize{}401} & {\scriptsize{}23599} & {\scriptsize{}596400}\tabularnewline
{\scriptsize{}Chicago} & {\scriptsize{}115} & {\scriptsize{}77} & {\scriptsize{}150} & {\scriptsize{}54} & {\scriptsize{}38} & {\scriptsize{}73} & {\scriptsize{}39949} & {\scriptsize{}318800}\tabularnewline
{\scriptsize{}San Francisco} & {\scriptsize{}57} & {\scriptsize{}38} & {\scriptsize{}81} & {\scriptsize{}24} & {\scriptsize{}16} & {\scriptsize{}36} & {\scriptsize{}19142} & {\scriptsize{}277700}\tabularnewline
{\scriptsize{}Washington} & {\scriptsize{}36} & {\scriptsize{}21} & {\scriptsize{}59} & {\scriptsize{}29} & {\scriptsize{}16} & {\scriptsize{}48} & {\scriptsize{}17253} & {\scriptsize{}254975}\tabularnewline
{\scriptsize{}Los Angeles} & {\scriptsize{}34} & {\scriptsize{}16} & {\scriptsize{}58} & {\scriptsize{}17} & {\scriptsize{}7} & {\scriptsize{}31} & {\scriptsize{}28816} & {\scriptsize{}336000}\tabularnewline
{\scriptsize{}Philadelphia} & {\scriptsize{}21} & {\scriptsize{}8} & {\scriptsize{}28} & {\scriptsize{}9} & {\scriptsize{}5} & {\scriptsize{}13} & {\scriptsize{}13640} & {\scriptsize{}192204}\tabularnewline
{\scriptsize{}Boston} & {\scriptsize{}18} & {\scriptsize{}8.3} & {\scriptsize{}26} & {\scriptsize{}12} & {\scriptsize{}8} & {\scriptsize{}17} & {\scriptsize{}11625} & {\scriptsize{}669000}\tabularnewline
{\scriptsize{}Houston} & {\scriptsize{}11} & {\scriptsize{}6.7} & {\scriptsize{}15} & {\scriptsize{}9} & {\scriptsize{}6} & {\scriptsize{}12} & {\scriptsize{}20979} & {\scriptsize{}308600}\tabularnewline
{\scriptsize{}Newark} & {\scriptsize{}7.3} & {\scriptsize{}0.8} & {\scriptsize{}12} & {\scriptsize{}4} & {\scriptsize{}0.1} & {\scriptsize{}9} & {\scriptsize{}12827} & {\scriptsize{}518100}\tabularnewline
{\scriptsize{}Seattle} & {\scriptsize{}6.7} & {\scriptsize{}4} & {\scriptsize{}10} & {\scriptsize{}4} & {\scriptsize{}3} & {\scriptsize{}6} & {\scriptsize{}13578} & {\scriptsize{}212000}\tabularnewline
\hline 
\end{tabular}
\par\end{centering}{\scriptsize \par}
\caption{Terrorism losses, average daily air departures, and average daily
bridge traffic}
\end{table}
We follow \cite{Hu2011} and assume that both daily air departures
and daily bridge traffic have a log-uniform distribution (the log-uniform
distribution is also used in \cite{Willis}). We let the random variable
$U$ have distribution:
\[
\mathbb{P}(U=-1)=\mathbb{P}(U=0)=\mathbb{P}(U=1)=1/3.
\]
For attributes $i=3,\,4$ (air departures and bridge traffic), we
set $\overline{T}_{ij}$ to be the \textit{averages} from columns
3 and 4 of Table 1. For a constant $\gamma>1$ (which controls volatility),
we set
\[
T_{ij}=\left(\frac{2\,\gamma\,\overline{T}_{ij}\ln\gamma}{(\gamma^{2}-1)}\right)\gamma^{U}
\]
to be the random number of daily incidents in each city $j$ for $i=3,\,4$.
Finally, we let $C_{ij}=c_{i}T_{ij}$ denote the random costs corresponding
to attributes $i=3,\,4$, where $c_{i}$ is the economic loss per
incident of attribute $i$. In our experiments, we set $\gamma=1.1$,
$c_{3}=\$500$, and $c_{4}=\$300$. Table 2 shows the random losses
for daily air departures (DAD) and daily bridge traffic (DBT) under
each scenario using this construction.

\begin{table}
\begin{centering}
{\scriptsize{}}%
\begin{tabular}{lcccccc}
\hline 
 & \multicolumn{3}{c}{\emph{\scriptsize{}DAD (\$ million)}} & \multicolumn{3}{c}{\emph{\scriptsize{}DBT (\$ million)}}\tabularnewline
\cline{2-7} 
\emph{\scriptsize{}Urban area} & \emph{\scriptsize{}Standard} & \emph{\scriptsize{}Reduced} & \emph{\scriptsize{}Increased} & \emph{\scriptsize{}Standard} & \emph{\scriptsize{}Reduced} & \emph{\scriptsize{}Increased}\tabularnewline
\hline 
{\scriptsize{}New York} & {\scriptsize{}10.71} & \multicolumn{1}{c}{{\scriptsize{}11.78}} & {\scriptsize{}12.96} & {\scriptsize{}162.41} & {\scriptsize{}178.65} & {\scriptsize{}196.51}\tabularnewline
{\scriptsize{}Chicago} & {\scriptsize{}18.13} & {\scriptsize{}19.94} & {\scriptsize{}21.94} & {\scriptsize{}86.81} & {\scriptsize{}95.50} & {\scriptsize{}105.04}\tabularnewline
{\scriptsize{}San Francisco} & {\scriptsize{}8.69} & {\scriptsize{}9.56} & {\scriptsize{}10.51} & {\scriptsize{}75.62} & {\scriptsize{}83.18} & {\scriptsize{}91.50}\tabularnewline
{\scriptsize{}Washington} & {\scriptsize{}7.83} & {\scriptsize{}8.61} & {\scriptsize{}9.47} & {\scriptsize{}69.43} & {\scriptsize{}76.38} & {\scriptsize{}84.01}\tabularnewline
{\scriptsize{}Los Angeles} & {\scriptsize{}13.08} & {\scriptsize{}14.39} & {\scriptsize{}15.82} & {\scriptsize{}91.50} & {\scriptsize{}100.65} & {\scriptsize{}110.71}\tabularnewline
{\scriptsize{}Philadelphia} & {\scriptsize{}6.19} & {\scriptsize{}6.81} & {\scriptsize{}7.49} & {\scriptsize{}52.34} & {\scriptsize{}57.57} & {\scriptsize{}63.3314}\tabularnewline
{\scriptsize{}Boston} & {\scriptsize{}5.28} & {\scriptsize{}5.80} & {\scriptsize{}6.38} & {\scriptsize{}182.18} & {\scriptsize{}200.40} & {\scriptsize{}220.44}\tabularnewline
{\scriptsize{}Houston} & {\scriptsize{}9.52} & {\scriptsize{}10.47} & {\scriptsize{}11.52} & {\scriptsize{}84.04} & {\scriptsize{}92.44} & {\scriptsize{}101.68}\tabularnewline
{\scriptsize{}Newark} & {\scriptsize{}5.82} & {\scriptsize{}6.40} & {\scriptsize{}7.04} & {\scriptsize{}141.09} & {\scriptsize{}155.19} & {\scriptsize{}170.71}\tabularnewline
{\scriptsize{}Seattle} & {\scriptsize{}6.16} & {\scriptsize{}6.78} & {\scriptsize{}7.46} & {\scriptsize{}57.37} & {\scriptsize{}63.50} & {\scriptsize{}69.85}\tabularnewline
\hline 
\end{tabular}
\par\end{centering}{\scriptsize \par}
\caption{Air departures and bridge traffic}
\end{table}
The data in Tables 1 and 2 are used to form the random loss matrix
$C$.

\subsection{The Experiments}

In this section, we describe the details of our experiments and present
the results. We start by describing the preference elicitation procedure,
where we simulate the preferences of DHS by using the multivariate
certainty equivalent. First, we take a continuous, monotonically increasing,
and concave univariate utility function $u\text{ : }\mathbb{R}\rightarrow\mathbb{R}$
with $u(0)=0$. Then, we choose weights $w\in\mathbb{R}^{n}$ with
$w\geq0$ and $\|w\|_{2}=1$ to define the multivariate certainty
equivalent
\begin{equation}
\rho(X)=u^{-1}\left\{ \mathbb{E}\left[u\left(\left\langle w,\,-X\right\rangle \right)\right]\right\} ,\,X\in\mathcal{L},\label{MCE}
\end{equation}
where $u^{-1}$ is the inverse of $u$. Here, we use the exponential
utility function $u(y)=-\exp(-\kappa\,y)$ with $\kappa=0.05$. 

The data set used to elicit preferences consists of many pairs of
prospects. DHS is asked to choose the preferred prospect from each
pair of prospects (the preferred prospect is denoted $W_{i}$ and
the other is denoted $Y_{i}$), and the choice is made using the choice
function (\ref{MCE}). We refer to this data set as the ``elicited
comparison data set''.
\begin{rem}
We emphasi\textcolor{black}{ze that we select the choice function
(\ref{MCE}) and the value of $\kappa$ artificially f}or our experiments.
It is used to elicit consistent preferences on the elicited comparison
data set. Our robust preference model (\ref{DHS}) and our algorithm
do not have knowledge of this choice function.
\end{rem}

In the experiments, we set the overall budget to be $B=\$\,400$ million
and the effectiveness ratio of investment to be $\delta=0.05$ (these
are the same settings as in \cite{Nikoofal2012}).\textcolor{red}{{}
}For computation tractability, we use a piecewise linear function
to approximate the original effectiveness function (\ref{effectiveness}).
This piecewise linear function is constructed from a set of segments
with ten breakpoints $\{5,\,10,\,20,\,30,\,40,\,50,\,75,\,100,\,150,\,200\}$. 

We assume that the budget allocated to each city $j$ corresponding
to attribute $i$ cannot be less than \$$1$ million, i.e. $z_{ij}\geq1$
for all $i=1,\ldots,\,n$ and $j=1,\ldots,\,m$. In Problem (\ref{MILP})
- (\ref{MILP-10}), we take the modulus of Lipschitz continuity to
be $L=1/12$ (this choice is without loss of generality since the
robust choice function $\psi\left(\cdot;\,\mathcal{S},\,Y\right)$
will preserve the same preferences even if we change the modulus of
continuity of the choice functions in $\mathcal{S}$). For implementation
of Algorithm \ref{alg:Level}, we set $\lambda=0.8$ and $\epsilon=0.0001$.
It should be noted that, when solving Problem (\ref{DHS}), the value
of the function $\psi\left(-C\left(\cdot\right);\,\mathcal{S},\,Y\right)$
is obtained by solving the MILP (\ref{MILP})-(\ref{MILP-10}). The
optimal solution from Problem (\ref{discretize}) in Algorithm \ref{alg:Level}
is obtained by solving a sequence of MILPs. Finally, we set the DHS
targets $v_{ij}$ for $i=1,\ldots,\,n$ and $j=1,\ldots,\,m$ as shown
in Table 3. These targets vary from city to city since they are related
to the economic condition of the corresponding city.

\begin{table}
\centering{}{\scriptsize{}}%
\begin{tabular}{lc}
\hline 
 & \multicolumn{1}{c}{\emph{\scriptsize{}Targets }}\tabularnewline
\emph{\scriptsize{}Urban area} & \emph{\scriptsize{}(\$ million)}\tabularnewline
\hline 
{\scriptsize{}New York} & {\footnotesize{}500}\tabularnewline
{\scriptsize{}Chicago} & {\footnotesize{}450}\tabularnewline
{\scriptsize{}San Francisco} & {\footnotesize{}400}\tabularnewline
{\scriptsize{}Washington} & {\footnotesize{}350}\tabularnewline
{\scriptsize{}Los Angeles} & {\footnotesize{}300}\tabularnewline
{\scriptsize{}Philadelphia} & {\footnotesize{}250}\tabularnewline
{\scriptsize{}Boston} & {\footnotesize{}200}\tabularnewline
{\scriptsize{}Houston} & {\footnotesize{}150}\tabularnewline
{\scriptsize{}Newark} & {\footnotesize{}100}\tabularnewline
{\scriptsize{}Seattle} & {\footnotesize{}50}\tabularnewline
\hline 
\end{tabular}\caption{Targets}
\end{table}

\subsubsection{Experiment I: Risk-neutral vs. Robust preference}

In our first experiment, we compare the solutions of Problems (\ref{NEUTRAL})
and (\ref{DHS}). The elicited comparison data set used for this experiment
is presented in Appendix B, it contains five pairs of prospects in
total. We define these prospects over three scenarios where each scenario
has an equal one-third probability of being realized. For each prospect,
we generate the random loss level from the uniform distribution on
$[-1000,\,0${]}\textcolor{black}{. In this experiment, Algorithm
\ref{alg:Level} terminates in iteration $i=12$ when it finds the
optimal solution of Problem (\ref{DHS}). W}hen solving Problems (\ref{NEUTRAL})
and (\ref{DHS}), we obtain the optimal UASI budget allocations as
shown in following tables. Table 4 gives the optimal risk-neutral
UASI budget allocation and Table 5 gives the optimal UASI budget allocation
for our robust choice model.

\begin{table}[t]
\begin{centering}
{\footnotesize{}}%
\begin{tabular}{lcccc}
\hline 
 & \multicolumn{1}{c}{\emph{\footnotesize{}Property losses }} & \multirow{1}{*}{\emph{\footnotesize{}Fatalities}} & \emph{\footnotesize{}Air departures} & \emph{\footnotesize{}Average daily bridge traffic}\tabularnewline
\emph{\footnotesize{}Urban area} & \emph{\footnotesize{}(\$ million)} & \emph{\footnotesize{} (\$ million)} & \emph{\footnotesize{}(\$ million)} & \emph{\footnotesize{}(\$ million)}\tabularnewline
\hline 
{\scriptsize{}New York} & {\footnotesize{}62.472} & {\footnotesize{}32.253} & {\footnotesize{}1.000} & {\footnotesize{}9.985}\tabularnewline
{\scriptsize{}Chicago} & {\footnotesize{}8.017} & {\footnotesize{}3.485} & {\footnotesize{}1.000} & {\footnotesize{}5.314}\tabularnewline
{\scriptsize{}San Francisco} & {\footnotesize{}4.520} & {\footnotesize{}1.631} & {\footnotesize{}1.000} & {\footnotesize{}5.194}\tabularnewline
{\scriptsize{}Washington} & {\footnotesize{}3.648} & {\footnotesize{}2.841} & {\footnotesize{}1.000} & {\footnotesize{}5.484}\tabularnewline
{\scriptsize{}Los Angeles} & {\footnotesize{}4.284} & {\footnotesize{}1.973} & {\footnotesize{}1.000} & {\footnotesize{}9.194}\tabularnewline
{\scriptsize{}Philadelphia} & {\footnotesize{}2.196} & {\footnotesize{}1.000} & {\footnotesize{}1.000} & {\footnotesize{}5.825}\tabularnewline
{\scriptsize{}Boston} & {\footnotesize{}2.658} & {\footnotesize{}1.502} & {\footnotesize{}1.000} & {\footnotesize{}59.961}\tabularnewline
{\scriptsize{}Houston} & {\footnotesize{}1.888} & {\footnotesize{}1.374} & {\footnotesize{}1.000} & {\footnotesize{}22.488}\tabularnewline
{\scriptsize{}Newark} & {\footnotesize{}2.401} & {\footnotesize{}1.631} & {\footnotesize{}1.000} & {\footnotesize{}59.961}\tabularnewline
{\scriptsize{}Seattle} & {\footnotesize{}4.456} & {\footnotesize{}2.401} & {\footnotesize{}1.000} & {\footnotesize{}59.961}\tabularnewline
\hline 
\end{tabular}
\par\end{centering}{\footnotesize \par}
\caption{Optimal UASI budget allocation: risk-neutral}
\end{table}

\begin{table}[t]
\begin{centering}
{\footnotesize{}}%
\begin{tabular}{lcccc}
\hline 
 & \multicolumn{1}{c}{\emph{\footnotesize{}Property losses }} & \multirow{1}{*}{\emph{\footnotesize{}Fatalities}} & \emph{\footnotesize{}Air departures} & \emph{\footnotesize{}Average daily bridge traffic}\tabularnewline
\emph{\footnotesize{}Urban area} & \emph{\footnotesize{}(\$ million)} & \emph{\footnotesize{} (\$ million)} & \emph{\footnotesize{}(\$ million)} & \emph{\footnotesize{}(\$ million)}\tabularnewline
\hline 
{\scriptsize{}New York} & {\footnotesize{}78.829} & {\footnotesize{}32.253} & {\footnotesize{}1.000} & {\footnotesize{}9.985}\tabularnewline
{\scriptsize{}Chicago} & {\footnotesize{}8.017} & {\footnotesize{}3.485} & {\footnotesize{}1.000} & {\footnotesize{}5.314}\tabularnewline
{\scriptsize{}San Francisco} & {\footnotesize{}4.520} & {\footnotesize{}1.631} & {\footnotesize{}1.000} & {\footnotesize{}5.194}\tabularnewline
{\scriptsize{}Washington} & {\footnotesize{}3.648} & {\footnotesize{}2.841} & {\footnotesize{}1.000} & {\footnotesize{}5.484}\tabularnewline
{\scriptsize{}Los Angeles} & {\footnotesize{}4.284} & {\footnotesize{}1.973} & {\footnotesize{}1.000} & {\footnotesize{}9.195}\tabularnewline
{\scriptsize{}Philadelphia} & {\footnotesize{}2.196} & {\footnotesize{}1.000} & {\footnotesize{}1.000} & {\footnotesize{}5.825}\tabularnewline
{\scriptsize{}Boston} & {\footnotesize{}2.658} & {\footnotesize{}1.502} & {\footnotesize{}1.000} & {\footnotesize{}59.961}\tabularnewline
{\scriptsize{}Houston} & {\footnotesize{}1.888} & {\footnotesize{}1.374} & {\footnotesize{}1.000} & {\footnotesize{}22.488}\tabularnewline
{\scriptsize{}Newark} & {\footnotesize{}2.401} & {\footnotesize{}1.631} & {\footnotesize{}1.000} & {\footnotesize{}59.961}\tabularnewline
{\scriptsize{}Seattle} & {\footnotesize{}4.456} & {\footnotesize{}2.401} & {\footnotesize{}1.000} & {\footnotesize{}43.604}\tabularnewline
\hline 
\end{tabular}
\par\end{centering}{\footnotesize \par}
\caption{Optimal UASI budget allocation: robust choice}
\end{table}
We test the performance of these two budget allocation plans in simulation
with $1000$ i.i.d samples. Figure 1 shows the histogram of costs
after $1000$ simulations for both the risk-neutral and robust choice
budgets. From Figure 1, we observe that the loss distribution of the
optimal budget allocation from Problem (\ref{DHS}) generally has
a slightly higher mean but lower variance compared to the optimal
budget allocation from Problem (\ref{NEUTRAL}). To make the comparison
more apparent, we fit a normal distribution to the simulation output.
Figure 2 shows that the random loss distribution of the optimal budget
allocation from our robust preference model second-order stochastically
dominates the random loss distribution of the optimal budget allocation
of the risk-neutral model. Let $F_{R}$ and $F_{N}$ denote the cumulative
distribution function of random loss distribution of the optimal budges
from the robust preference and risk-neutral models, respectively.
We can explicitly verify that second-order stochastic dominance is
satisfied by verifying that $\int_{-\infty}^{x}\left|F_{R}(t)-F_{N}(t)\right|dt\geq0.0288$
for all real numbers $x$. 

We also estimate the mean and variance of the loss distribution, as
well as compute its $5\%$ and $10\%$-level CVaR, under these two
optimal budget allocation plans. The results in Table 6 show explicitly
that the optimal budget allocation from Problem (\ref{DHS}) has a
slightly higher mean but that it also has a lower variance, $5\%$-CVaR,
and $10\%$-CVaR compared to the optimal budget allocation from Problem
(\ref{NEUTRAL}). This table also reveals that the budget allocation
from our robust choice model induces a lower probability of suffering
extremely high losses. Thus, we claim that the optimal budget allocation
from Problem (\ref{DHS}) is more resilient against loss from terror
attacks. These experimental results are in line with what we expect
from the diversification favoring behavior of quasi-concave choice
functions.

\begin{figure}
\begin{centering}
\includegraphics[scale=0.6]{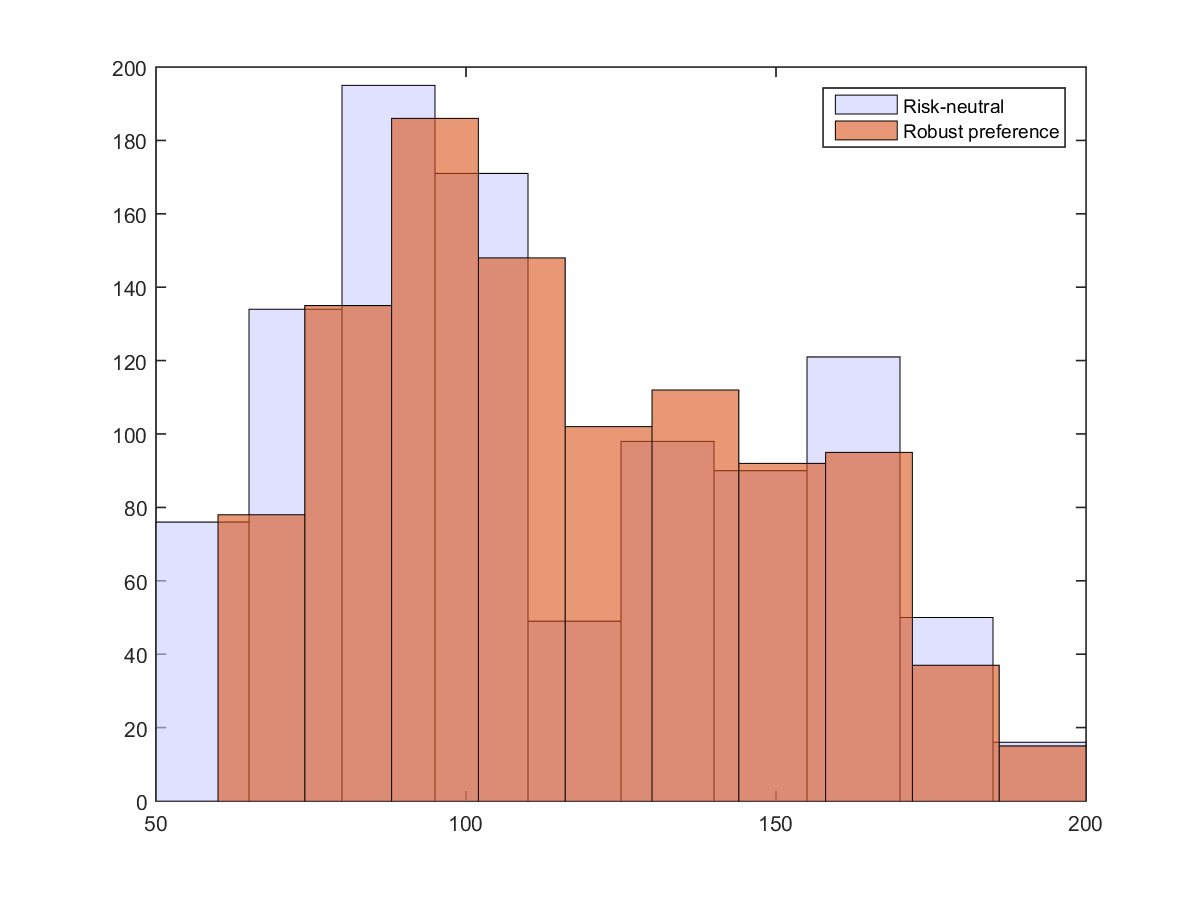}
\par\end{centering}
\begin{centering}
\includegraphics[scale=0.6]{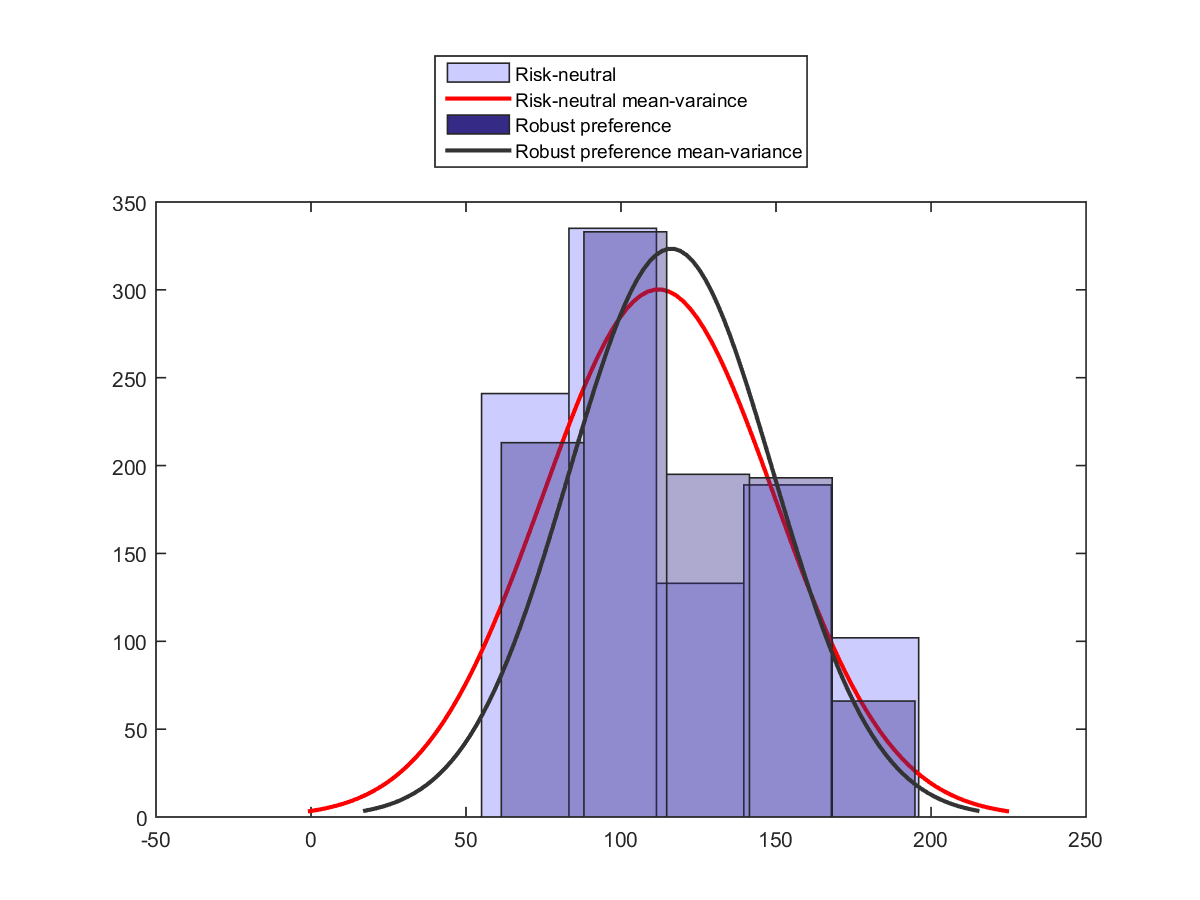}
\par\end{centering}
\caption{Simulation results}
\end{figure}

\begin{table}[t]
\begin{centering}
{\small{}}%
\begin{tabular}{lcccc}
\hline 
\emph{\small{}Budget allocation} & \emph{\small{}Mean (\$ million)} & \emph{\small{}Variance } & \emph{\small{}10\%-CVaR (\$ million)} & \emph{\small{}5\%-CVaR (\$ million)}\tabularnewline
\hline 
{\small{}Risk-neutral } & {\small{}112.088} & {\small{}1404.362} & {\small{}177.273} & {\small{}183.806}\tabularnewline
{\small{}Robust preference} & {\small{}116.213} & {\small{}1083.492} & {\small{}174.264} & {\small{}180.884}\tabularnewline
\hline 
\end{tabular}
\par\end{centering}{\small \par}
\caption{Optimal UASI budget allocation: robust choice}
\end{table}

\begin{figure}
\begin{centering}
\includegraphics[scale=0.5]{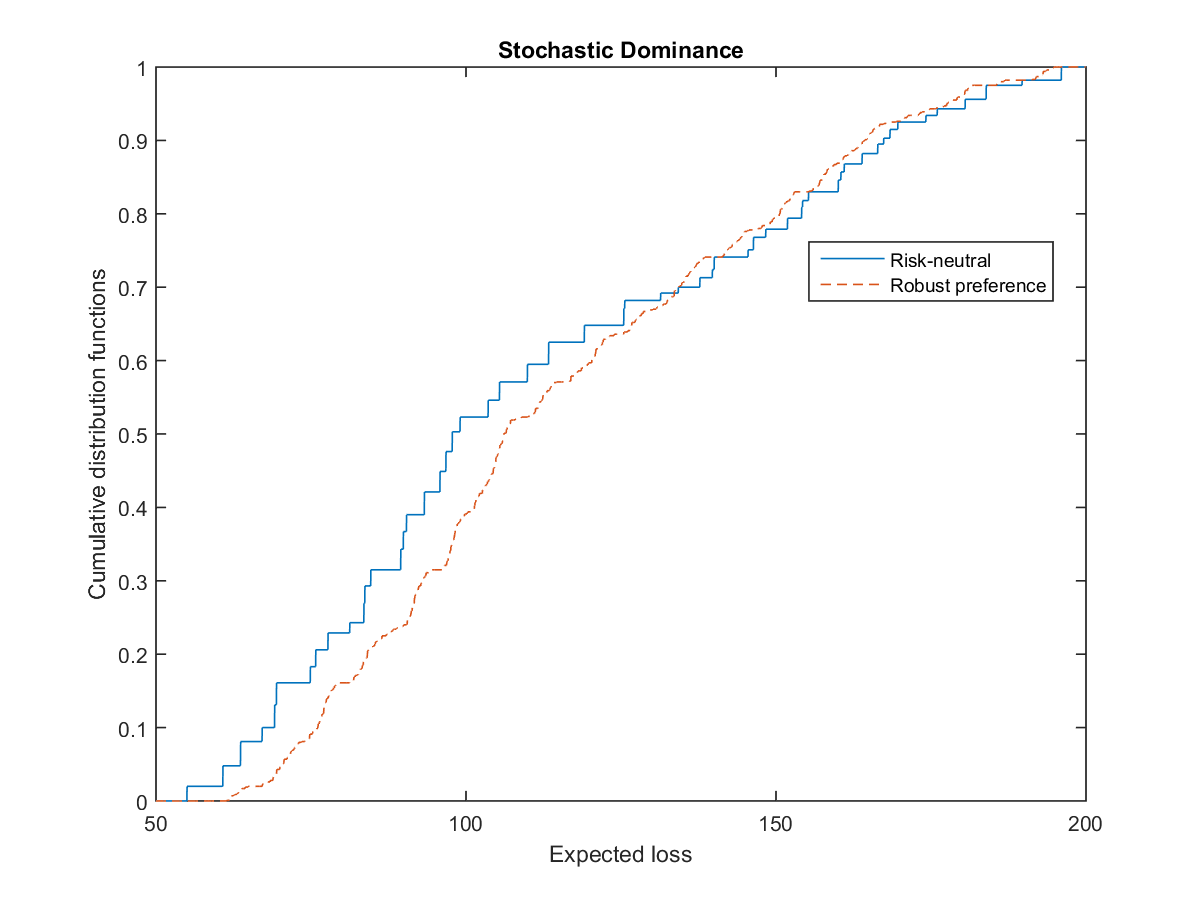}
\par\end{centering}
\caption{Stochastic dominance}
\end{figure}

\subsubsection{Experiment II: Sensitivity analysis}

In our second experiment, we conduct sensitivity analysis on the effectiveness
ratio $\delta$ in (\ref{effectiveness}) in our robust choice model.
Also, we study how sensitive the optimal budget allocation is to different
elicited comparison data sets. In this experiment, we fix $\kappa=0.05$.

Table 7 shows that as the effectiveness ratio $\delta$ increases,
the optimal budget allocation becomes less sensitive, and the expected
loss decreases. One possible explanation is that if $\delta$ is extremely
low, putting all of the budget in a single city and on a single attribute
cannot achieve an investment close to the corresponding economic target.
However, when $\delta$ is sufficiently large, the investment triggered
by the budget allocated in each city on each attribute is able to
achieve its economic target. In this case, the optimal budget allocation
will be insensitive to changes in $\delta$ and the expected loss
becomes stable. 

Table 8 shows how sensitive the optimal budget allocation is to different
elicited comparison data sets in our robust choice model. In this
experiment, we fix $\kappa=\delta=0.05$. We run the experiments for
twenty different elicited comparison data sets with the same number
of pairs. The ``variance of the optimal budget allocation'' is computed
by summing the variances of the optimal budget allocations in each
city on each attribute among twenty groups. The ``variance of expected
loss'' is the variance of the expected loss from the optimal budget
allocation among twenty groups. As the size of elicited comparison
data set increases, these two variance measures will increase first
and then decrease. We hypothesize that when the number of pairs is
small, the preference estimation error is high. When the size of the
elicited comparison data set is large, the preference estimation error
is reduced and differences in elicited comparison data sets will not
affect the optimal budget allocation. 

\begin{table}
\begin{centering}
{\footnotesize{}}%
\begin{tabular}{ccccc}
\hline 
{\footnotesize{}Effectiveness ratio $\delta$} & {\footnotesize{}0.01} & {\footnotesize{}0.02} & {\footnotesize{}0.04} & {\footnotesize{}0.05}\tabularnewline
\hline 
\emph{\footnotesize{}Expected loss (\$ million)} & \textcolor{black}{\footnotesize{}1034.162} & \textcolor{black}{\footnotesize{}483.320} & \textcolor{black}{\footnotesize{}145.784} & \textcolor{black}{\footnotesize{}111.525}\tabularnewline
\hline 
{\footnotesize{}Effectiveness ratio $\delta$} & \textcolor{black}{\footnotesize{}0.1} & \textcolor{black}{\footnotesize{}0.2} & \textcolor{black}{\footnotesize{}0.5} & \textcolor{black}{\footnotesize{}0.8}\tabularnewline
\hline 
\emph{\footnotesize{}Expected loss (\$ million)} & \textcolor{black}{\footnotesize{}93.059} & \textcolor{black}{\footnotesize{}92.973} & \textcolor{black}{\footnotesize{}92.972} & \textcolor{black}{\footnotesize{}92.972}\tabularnewline
\hline 
\end{tabular}
\par\end{centering}{\footnotesize \par}
\caption{Sensitivity analysis $\delta$}
\end{table}

\begin{table}
\begin{centering}
{\footnotesize{}}%
\begin{tabular}{ccccccc}
\hline 
{\footnotesize{}Number of pairs} & {\footnotesize{}1} & {\footnotesize{}3} & {\footnotesize{}5} & {\footnotesize{}10} & {\footnotesize{}25} & {\footnotesize{}50}\tabularnewline
\hline 
\emph{\footnotesize{}Variance of optimal budget allocation} & {\footnotesize{}2.45E-29 } & {\footnotesize{}2.45E-29 } & {\footnotesize{}3.50E-07 } & {\footnotesize{}3.664429998 } & {\footnotesize{}6.38E-06 } & {\footnotesize{}4.87E-08 }\tabularnewline
\hline 
\emph{\footnotesize{}Variance of expected loss} & {\footnotesize{}8.50E-28 } & {\footnotesize{}8.50E-28 } & {\footnotesize{}3.11E-10 } & {\footnotesize{}0.139151122 } & {\footnotesize{}5.80E-09 } & {\footnotesize{}2.23E-11}\tabularnewline
\hline 
\end{tabular}
\par\end{centering}{\footnotesize \par}
\caption{Sensitivity analysis: elicited comparison data set}
\end{table}

\subsubsection{Experiment III: Out-of-sample performance}

In our third experiment, we conduct out-of-sample tests on the effectiveness
of our robust choice model (\ref{MILP})-(\ref{MILP-10}), in terms
of eliciting the true preferences of the decision maker, when the
true choice function is (\ref{MCE}). Two data sets are constructed
for this experiment: one is the elicited comparison data set with
$100$ pairs of prospects, the other is a test data set with twenty
pairs of prospects. Both data sets are generated following the same
procedures as for Experiment I. In each run, we select the first $I$
pairs of prospects from all $100$ pairs, 
\begin{align*}
I\in & \{1,\,2,\,3,\,4,\,5,\,6,\,10,\,15,\,20,\,25,\\
 & 30,\,35,\,40,\,45,\,50,\,60,\,75,\,100\},
\end{align*}
to form the elicited comparison data set, and then we solve Problem
(\ref{MILP})-(\ref{MILP-10}) to compare the robust preference of
each pair of prospects in the test data set. We compare the robust
preference with the true preference derived by the true choice function
(\ref{MCE}), and record the number of violated orders in the test
data set. Table 9 shows that as $|I|$ increases, the number of violated
orders in the test data set diminishes. This phenomenon confirms that,
as $|I|$ increases, the model (\ref{MILP})-(\ref{MILP-10}) with
extra preference information provides a better estimate of the preferences
induced by the true choice function (\ref{MCE}). 

\begin{table}
\begin{centering}
{\footnotesize{}}%
\begin{tabular}{cccccccccc}
\hline 
{\footnotesize{}Number of pairs} & {\footnotesize{}1} & {\footnotesize{}2} & {\footnotesize{}3} & {\footnotesize{}4} & {\footnotesize{}5} & {\footnotesize{}6} & {\footnotesize{}10} & {\footnotesize{}15} & {\footnotesize{}20}\tabularnewline
\hline 
\emph{\footnotesize{}Number of violated orders} & {\footnotesize{}7} & {\footnotesize{}7} & {\footnotesize{}7} & {\footnotesize{}7} & {\footnotesize{}7} & {\footnotesize{}6} & {\footnotesize{}5} & {\footnotesize{}4} & {\footnotesize{}4}\tabularnewline
\hline 
{\footnotesize{}Number of pairs} & {\footnotesize{}25} & {\footnotesize{}30} & {\footnotesize{}35} & {\footnotesize{}40} & {\footnotesize{}45} & {\footnotesize{}50} & {\footnotesize{}60} & {\footnotesize{}75} & {\footnotesize{}100}\tabularnewline
\hline 
\emph{\footnotesize{}Number of violated orders} & {\footnotesize{}5} & {\footnotesize{}6} & {\footnotesize{}2} & {\footnotesize{}4} & {\footnotesize{}4} & {\footnotesize{}4} & {\footnotesize{}4} & {\footnotesize{}3} & {\footnotesize{}2}\tabularnewline
\hline 
\end{tabular}
\par\end{centering}{\footnotesize \par}
\caption{Out-of-sample performance}
\end{table}

\subsubsection{Experiment IV: Runtime}

In our fourth experiment, we study how the solution time of Problem
(\ref{DHS}) depends on the size of the elicited comparison data set.
We consider preference elicitation with $I\in\{1,\,2,\,4,\,6,\,8,\,10\}$
pairs of prospects. For each $I$, we generate the appropriate elicited
comparison data set following the same procedure as in Experiment
I, and then we solve the corresponding instance of Problem (\ref{DHS}).
The experiments were performed on a generic laptop with Intel Core
i7 processor, 8GM RAM, on a 64-bit Windows 8 operating system via
Matlab R2015a and CPLEX Studio 12.5. Table 10 and Figure 2 show that
the runtime of Problem (\ref{DHS}) grows gracefully as a function
of the size of the elicited comparison data set.

\begin{table}
\begin{centering}
{\footnotesize{}}%
\begin{tabular}{ccccccc}
\hline 
{\footnotesize{}Number of pairs} & {\footnotesize{}1} & {\footnotesize{}2} & {\footnotesize{}4} & {\footnotesize{}6} & {\footnotesize{}8} & {\footnotesize{}10}\tabularnewline
\hline 
{\footnotesize{}Time (second)} & \textcolor{black}{\footnotesize{}117.743} & \textcolor{black}{\footnotesize{}148.583 } & \textcolor{black}{\footnotesize{}235.051 } & \textcolor{black}{\footnotesize{}460.283} & \textcolor{black}{\footnotesize{}1665.427 } & \textcolor{black}{\footnotesize{}6913.815 }\tabularnewline
\hline 
\end{tabular}
\par\end{centering}{\footnotesize \par}
\caption{Computation time}
\end{table}

\begin{figure}
\begin{centering}
\includegraphics[scale=0.25]{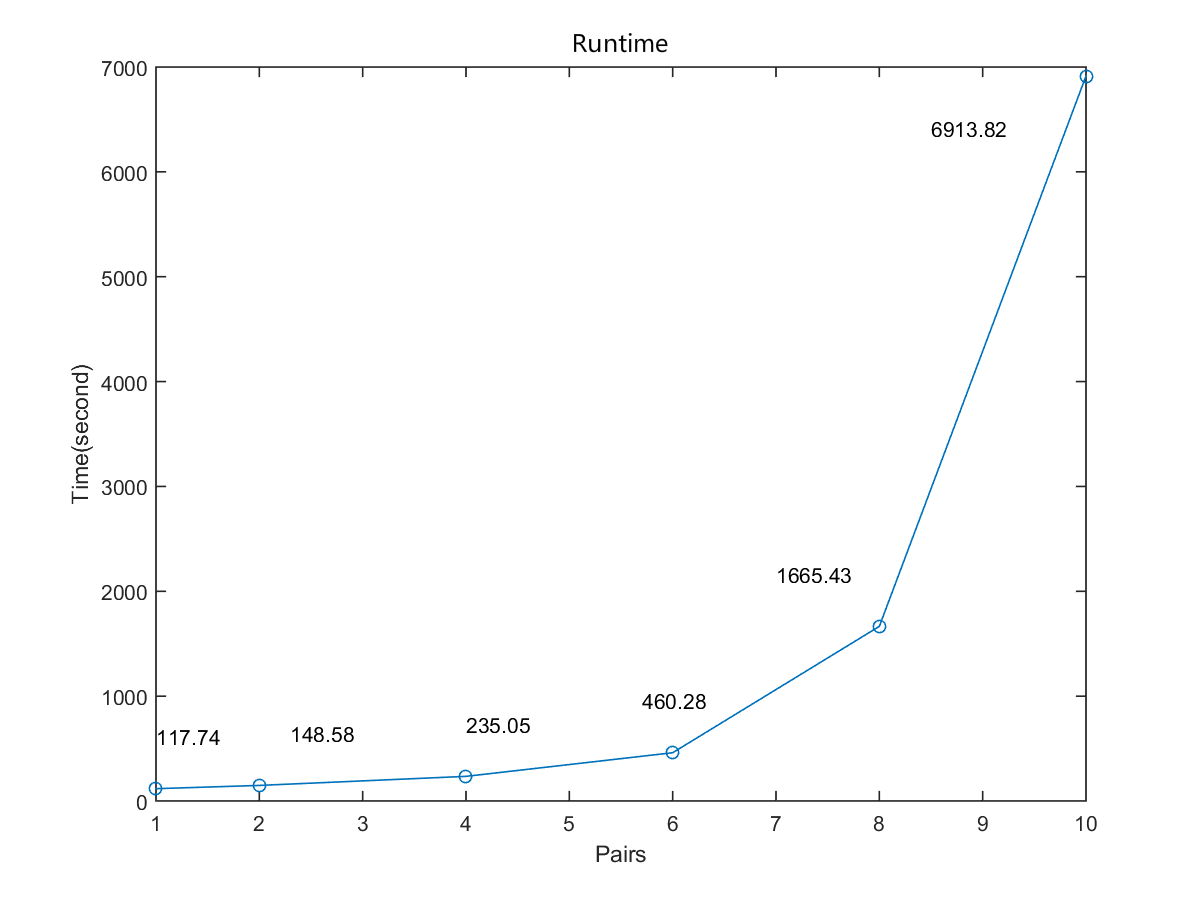}
\par\end{centering}
\caption{Runtime}
\end{figure}

\section{Conclusion}

This paper attempts to take on the simultaneous challenges of multi-attribute
prospects and non-convexity in robust utility/risk choice models.
It advances the existing research of preference robust optimization
with a more general mathematical model which complements the existing
convex optimization models, tractable computational schemes, and structural
insights for modern decision-making problems.

As our first main result, we show how to optimize our robust choice
model using support functions. We give a tractable procedure based
on solving a sequence of MILPs. As our second main result, we give
a common representation for multi-attribute quasi-concave choice functions.
This representation result shows that all multi-attribute quasi-concave
choice functions can be expressed in terms of a family of convex risk
functions and satiation levels (a.k.a. targets). This result is analogous
to the main result in \cite{brown2012aspirational}. We develop a
case study in homeland security which shows that our model is both:
(i) highly expressive and can incorporate decision maker preference
information; and (ii) provides solutions with more favorable risk
properties compared to other models.

In future research, we will consider further computational issues
of our model, in particular application to general probability spaces
(where the sample space is not necessarily finite). In parallel, we
will explore the connection between our present work and multi-attribute
stochastic dominance models.
\begin{acknowledgement*}
This research is supported by MOE Tier II grant MOE2014-T2-1-066 and
MOE Tier I grant WBS R-266-000-083-133. The authors are grateful to
Erick Delage, Jonathan Yu-Meng Li, and Melvyn Sim for valuable comments
and suggestions during the preparation of this paper.
\end{acknowledgement*}
 \bibliographystyle{plain}
\bibliography{References}

\appendix
%dummy comment inserted by tex2lyx to ensure that this paragraph is not empty

\section{Basic results}

The following proposition is used to show that several of the risk
measures constructed in this paper are convex. 
\begin{prop}
\label{prop:Basic_risk} Let $\rho\in\mathcal{R}_{iqv}$, and for
$d\geq0$, define $\mu_{k}\text{ : }\mathcal{L}\rightarrow\mathbb{R}$
for all $k\in\mathbb{R}$ by 
\[
\mu_{k}\left(X\right)\triangleq\inf\left\{ \alpha\in\mathbb{R}\text{ : }\rho\left(X+\alpha\,d\right)\geq k\right\} ,\,\forall X\in\mathcal{L}.
\]
If $\rho\left(X+\alpha\,d\right)$ is increasing w.r.t. $\alpha$,
then the following hold. 
\begin{itemize}
\item[(i)] For each $X\in\mathcal{L}$, 
\[
\mu_{k}(X)\leq0\Longleftrightarrow\rho(X)\leq0;
\]
\item[(ii)] $\mu_{k}(\cdot)$ is monotone and convex over $\mathcal{L}$ and
has closed acceptance sets; 
\item[(iii)] $\mu_{k}(\cdot)$ is monotonically increasing in $k$. 
\end{itemize}
\end{prop}

\begin{proof}
Part (i). Let 
\[
\mathcal{F}(X,\,k)\triangleq\left\{ \alpha\in\mathbb{R}\text{ : }\rho\left(X+\alpha\,d\right)\geq k\right\} .
\]
If $\rho(X)\leq k$, then $0\in\mathcal{F}(X,\,k)$ and consequently
$\mu_{k}(X)\leq0$. Conversely, if $\mu_{k}(X)\leq0$, then for any
positive number $\epsilon$, the increasing property of $\rho\left(X+\alpha\,d\right)$
in $\alpha$ implies $0+\epsilon\in\mathcal{F}(X,\,k)$, that is,
\[
\rho\left(X+\epsilon\,d\right)\geq k.
\]
Driving $\epsilon$ to $0$, we deduce $\rho\left(X\right)\geq k$.

Part (ii).The closedness of the acceptance sets follows from Part
(i) and closedness of the upper level set of $\rho(\cdot)$. So we
are left to prove monotonicity and convexity.

For any $X,\,Y\in\mathcal{L}$ with $X\leq Y$, since $\rho(\cdot)$
is non-decreasing, 
\[
\rho\left(X+\alpha\,d\right)\leq\rho\left(Y+\alpha\,d\right),\forall\alpha\in\mathbb{R},
\]
which implies $\mathcal{F}(X,\,k)\subset\mathcal{F}(Y,\,k)$ and subsequently
$\mu_{k}\left(X\right)\geq\mu_{k}\left(Y\right)$. Moreover, for $\lambda\in\left[0,\,1\right]$,
and arbitrarily small positive number $\epsilon>0$, it follows by
the definition of $\mu_{k}$ 
\[
\rho\left(X+\left(\mu_{k}\left(X\right)+\epsilon\right)d\right)\geq k\text{ and }\rho\left(Y+\left(\mu_{k}\left(Y\right)+\epsilon\right)d\right)\geq k.
\]
Let $a_{\lambda}=\lambda\,\mu_{k}\left(X\right)+\left(1-\lambda\right)\mu_{k}\left(Y\right)$.
Then 
\begin{align*}
\rho\left(\lambda\,X+\left(1-\lambda\right)Y+\left(a_{\lambda}+\epsilon\right)d\right)=\, & \rho\left(\lambda\left(X+\left(\mu_{k}\left(X\right)+\epsilon\right)d\right)+\left(1-\lambda\right)\left(Y+\left(\mu_{k}\left(Y\right)+\epsilon\right)d\right)\right)\\
\geq\, & \min\left\{ \rho\left(X+\left(\mu_{k}\left(X\right)+\epsilon\right)d\right),\,\rho\left(Y+\left(\mu_{k}\left(Y\right)+\epsilon\right)d\right)\right\} ,\\
\geq\, & k,
\end{align*}
where the first inequality follows from quasi-concavity of $\rho\in\mathcal{R}_{iqv}$,
and the second inequality follows from the definition of $\mu_{k}\left(X\right)$
and $\mu_{k}\left(Y\right)$. The above inequality then gives rise
to 
\begin{align*}
\mu_{k}\left(\lambda\,X+\left(1-\lambda\right)Y\right)=\, & \inf\left\{ a\in\mathbb{R}\text{ : }\rho\left(\lambda\,X+\left(1-\lambda\right)Y+a\,d\right)\geq k\right\} \\
\leq\, & a_{\lambda}\\
=\, & \lambda\,\mu_{k}\left(X\right)+\left(1-\lambda\right)\mu_{k}\left(Y\right),
\end{align*}
and the desired conclusion follows.

Part (iii). For any $X\in\mathcal{L}$, we must have $\mu_{k}\left(X\right)\leq\mu_{k'}\left(X\right)$
for $k\leq k'$ since $\rho\left(X+\left(\mu_{k'}\left(X\right)+\epsilon\right)d\right)\geq k'\geq k$
for all $\epsilon>0$. 
\end{proof}
The next proposition is used to show that several of the choice functions
constructed from risk measures in this paper belong to $\mathcal{R}_{iqv}$. 
\begin{prop}
\label{prop:Basic_choice} Let $\left\{ \mu_{k}\right\} _{k\in\mathbb{R}}$
be a class of convex risk measures. Assume that:

(i) all $\left\{ \mu_{k}\right\} _{k\in\mathbb{R}}$ have closed acceptance
sets;

(ii) for any $X\in\mathcal{L}$, $\mu_{k}\left(X\right)$ as a function
of $k$ is non-decreasing and left continuous;

(iii) for any $X\in\mathcal{L}$, there is $k\in\mathbb{R}$ such
that $\mu_{k}\left(X\right)\leq0$.\\
 Define $\vartheta\text{ : }\mathcal{L}\rightarrow\bar{\mathbb{R}}$
by 
\begin{equation}
\vartheta\left(X\right)\triangleq\sup\left\{ k\in\mathbb{R}\text{ : }\mu_{k}\left(X\right)\leq0\right\} ,\,\forall X\in\mathcal{L}.\label{eq:quasiconcavity}
\end{equation}
Then $\vartheta$ is upper semi-continuous, monotonic, and quasi-concave. 
\end{prop}

\begin{proof}
\textit{Upper semi-continuity:} We first show that $\vartheta$ is
upper semicontinuous, or equivalently that\\
 $\left\{ X\in\mathcal{L}\mbox{ : }\vartheta\left(X\right)\geq k\right\} $
is closed for all $k\in\mathbb{R}$. Let $\left\{ X_{i}\right\} _{i\geq0}\subset\left\{ X\in\mathcal{L}\mbox{ : }\vartheta\left(X\right)\geq k\right\} $
and suppose $\lim_{i\rightarrow\infty}X_{i}\rightarrow X$. Since
$\vartheta\left(X_{i}\right)\geq k$, it follows that $\mu_{k}\left(X_{i}\right)\leq0$
for all $i\geq0$ and thus the sequence $\left\{ X_{i}\right\} _{i\geq0}$
belongs to the acceptance set $\mathcal{A}_{\mu_{k}}\triangleq\left\{ X\in\mathcal{L}\text{ : }\mu_{k}\left(X\right)\leq0\right\} $.
Since $\mathcal{A}_{\mu_{k}}$ is closed we have $X\in\mathcal{A}_{\mu_{k}}$,
and so $\vartheta\left(X\right)\geq k$ and thus $X\in\left\{ X\in\mathcal{L}\mbox{ : }\vartheta\left(X\right)\geq k\right\} $.

\textit{Monotonicity:} Choose $X,\,Y\in\mathcal{L}$ with $X\leq Y$.
Each $\mu_{k}$ is monotonic, so if $\mu_{k}\left(X\right)\leq0$
then $\mu_{k}\left(Y\right)\leq0$ also holds, and thus $\vartheta\left(X\right)\leq\vartheta\left(Y\right)$.

\textit{Quasi-concavity:} Choose $X,\,Y\in\mathcal{L}$ and $\lambda\in\left[0,\,1\right]$.
Then we have 
\begin{align*}
\vartheta\left(\lambda\,X+\left(1-\lambda\right)Y\right)=\, & \sup\left\{ k\in\mathbb{R}\mbox{ : }\mu_{k}\left(\lambda\,X+\left(1-\lambda\right)Y\right)\leq0\right\} \\
\geq\, & \sup\left\{ k\in\mathbb{R}\mbox{ : }\lambda\,\mu_{k}\left(X\right)+\left(1-\lambda\right)\mu_{k}\left(Y\right)\leq0\right\} \\
\geq\, & \min\left\{ \vartheta\left(X\right),\,\vartheta\left(Y\right)\right\} ,
\end{align*}
where the first inequality uses convexity of $\left\{ \mu_{k}\right\} _{k\in\mathbb{R}}$,
and the second inequality uses the fact that $\lambda\,\mu_{k}\left(X\right)+\left(1-\lambda\right)\mu_{k}\left(Y\right)\geq0$
for $k^{*}=\min\left\{ \vartheta\left(X\right),\,\vartheta\left(Y\right)\right\} $
since both $\mu_{k}\left(X\right)\geq0$ and $\mu_{k}\left(Y\right)\geq0$
hold for $k^{*}$ by left continuity of $\mu_{k}\left(X\right)$ in
$k$ for all $X\in\mathcal{L}$. 
\end{proof}
The next proposition is used to establish our level set representation. 
\begin{prop}
\label{prop:Basic_level} Suppose 
\[
\left\{ X\in\mathcal{L}\text{ : }\rho\left(X\right)\geq k\right\} =\left\{ X\in\mathcal{L}\text{ : }\mu_{k}\left(X\right)\leq0\right\} ,\,\forall k\in\mathbb{R},
\]
then 
\[
\rho\left(X\right)=\sup\left\{ k\in\mathbb{R}\text{ : }\mu_{k}\left(X\right)\leq0\right\} ,\,\forall X\in\mathcal{L}.
\]
\end{prop}

\begin{proof}
By construction of $\left\{ \mu_{k}\right\} _{k\in\mathbb{R}}$, we
have: (i) if $\rho\left(X\right)\geq k$ then $\mu_{k}\left(X\right)\leq0$;
(ii) if $\mu_{k}\left(X\right)>0$ then $\rho\left(X\right)<k$. We
have $\mu_{k}\left(X\right)\leq0$ implies $\rho\left(X\right)\geq k$
for all $k$, so $\sup\left\{ k\in\mathbb{R}\text{ : }\mu_{k}\left(X\right)\leq0\right\} \leq\rho\left(X\right)$.
Conversely, we have $\mu_{\rho\left(X\right)}\left(X\right)\leq0$
and so $\sup\left\{ k\in\mathbb{R}\text{ : }\mu_{k}\left(X\right)\leq0\right\} \geq\rho\left(X\right)$. 
\end{proof}

\section{Elicited comparison data set}

We show the elicited comparison data set mentioned in Section 7.2
in the following tables. Losses are recorded here as negative values. 
\begin{center}
{\footnotesize{}{}}%
\begin{tabular}{cccc}
\hline 
{\footnotesize{}{}Attribute }\emph{\footnotesize{}{}} & \multicolumn{3}{c}{{\footnotesize{}{}Loss }\emph{\footnotesize{}{}(\$ million)}}\tabularnewline
\hline 
 & {\footnotesize{}{}Scenario I}  & {\footnotesize{}{}Scenario II}  & {\footnotesize{}{}Scenario III}\tabularnewline
\hline 
\emph{\footnotesize{}{}Property losses}{\footnotesize{} } & {\footnotesize{}{}-185}  & {\footnotesize{}{}-368}  & {\footnotesize{}{}-42}\tabularnewline
\emph{\footnotesize{}{}Fatalities}{\footnotesize{} } & {\footnotesize{}{}-94}  & {\footnotesize{}{}-903}  & {\footnotesize{}{}-35}\tabularnewline
\emph{\footnotesize{}{}Air Departures}{\footnotesize{} } & {\footnotesize{}{}-873}  & {\footnotesize{}{}-722}  & {\footnotesize{}{}-843}\tabularnewline
\emph{\footnotesize{}{}Average daily bridge traffic}{\footnotesize{} } & {\footnotesize{}{}-86}  & {\footnotesize{}{}-453}  & {\footnotesize{}{}-29}\tabularnewline
\hline 
\end{tabular}{\footnotesize{}{} }%
\begin{tabular}{cccc}
\hline 
{\footnotesize{}{}Attribute }\emph{\footnotesize{}{}} & \multicolumn{3}{c}{{\footnotesize{}{}Loss }\emph{\footnotesize{}{}(\$ million)}}\tabularnewline
\hline 
 & {\footnotesize{}{}Scenario I}  & {\footnotesize{}{}Scenario II}  & {\footnotesize{}{}Scenario III}\tabularnewline
\hline 
\emph{\footnotesize{}{}Property losses}{\footnotesize{} } & {\footnotesize{}{}-917}  & {\footnotesize{}{}-529}  & {\footnotesize{}{}-331}\tabularnewline
\emph{\footnotesize{}{}Fatalities}{\footnotesize{} } & {\footnotesize{}{}-25}  & {\footnotesize{}{}-439}  & {\footnotesize{}{}-794}\tabularnewline
\emph{\footnotesize{}{}Air Departures}{\footnotesize{} } & {\footnotesize{}{}-348}  & {\footnotesize{}{}-731}  & {\footnotesize{}{}-346}\tabularnewline
\emph{\footnotesize{}{}Average daily bridge traffic}{\footnotesize{} } & {\footnotesize{}{}-769}  & {\footnotesize{}{}-251}  & {\footnotesize{}{}-928}\tabularnewline
\hline 
\end{tabular}{\footnotesize{}{}}\\
 
\par\end{center}

\begin{center}
{\footnotesize{}{}}%
\begin{tabular}{cccc}
\hline 
{\footnotesize{}{}Attribute }\emph{\footnotesize{}{}} & \multicolumn{3}{c}{{\footnotesize{}{}Loss }\emph{\footnotesize{}{}(\$ million)}}\tabularnewline
\hline 
 & {\footnotesize{}{}Scenario I}  & {\footnotesize{}{}Scenario II}  & {\footnotesize{}{}Scenario III}\tabularnewline
\hline 
\emph{\footnotesize{}{}Property losses}{\footnotesize{} } & {\footnotesize{}{}-597}  & {\footnotesize{}{}-496}  & {\footnotesize{}{}-593}\tabularnewline
\emph{\footnotesize{}{}Fatalities}{\footnotesize{} } & {\footnotesize{}{}-878}  & {\footnotesize{}{}-353}  & {\footnotesize{}{}-333}\tabularnewline
\emph{\footnotesize{}{}Air Departures}{\footnotesize{} } & {\footnotesize{}{}-732}  & {\footnotesize{}{}-692}  & {\footnotesize{}{}-66}\tabularnewline
\emph{\footnotesize{}{}Average daily bridge traffic}{\footnotesize{} } & {\footnotesize{}{}-742}  & {\footnotesize{}{}-862}  & {\footnotesize{}{}-189}\tabularnewline
\hline 
\end{tabular}{\footnotesize{}{} }%
\begin{tabular}{cccc}
\hline 
{\footnotesize{}{}Attribute }\emph{\footnotesize{}{}} & \multicolumn{3}{c}{{\footnotesize{}{}Loss }\emph{\footnotesize{}{}(\$ million)}}\tabularnewline
\hline 
 & {\footnotesize{}{}Scenario I}  & {\footnotesize{}{}Scenario II}  & {\footnotesize{}{}Scenario III}\tabularnewline
\hline 
\emph{\footnotesize{}{}Property losses}{\footnotesize{} } & {\footnotesize{}{}-669}  & {\footnotesize{}{}-524}  & {\footnotesize{}{}-515}\tabularnewline
\emph{\footnotesize{}{}Fatalities}{\footnotesize{} } & {\footnotesize{}{}-848}  & {\footnotesize{}{}-638}  & {\footnotesize{}{}-243}\tabularnewline
\emph{\footnotesize{}{}Air Departures}{\footnotesize{} } & {\footnotesize{}{}-652}  & {\footnotesize{}{}-212}  & {\footnotesize{}{}-583}\tabularnewline
\emph{\footnotesize{}{}Average daily bridge traffic}{\footnotesize{} } & {\footnotesize{}{}-879}  & {\footnotesize{}{}-219}  & {\footnotesize{}{}-28}\tabularnewline
\hline 
\end{tabular}{\footnotesize{}{}}\\
 
\par\end{center}

\begin{center}
{\footnotesize{}{}}%
\begin{tabular}{cccc}
\hline 
{\footnotesize{}{}Attribute }\emph{\footnotesize{}{}} & \multicolumn{3}{c}{{\footnotesize{}{}Loss }\emph{\footnotesize{}{}(\$ million)}}\tabularnewline
\hline 
 & {\footnotesize{}{}Scenario I}  & {\footnotesize{}{}Scenario II}  & {\footnotesize{}{}Scenario III}\tabularnewline
\hline 
\emph{\footnotesize{}{}Property losses}{\footnotesize{} } & {\footnotesize{}{}-115}  & {\footnotesize{}{}-331}  & {\footnotesize{}{}-12}\tabularnewline
\emph{\footnotesize{}{}Fatalities}{\footnotesize{} } & {\footnotesize{}{}-906}  & {\footnotesize{}{}-867}  & {\footnotesize{}{}-135}\tabularnewline
\emph{\footnotesize{}{}Air Departures}{\footnotesize{} } & {\footnotesize{}{}-70}  & {\footnotesize{}{}-979}  & {\footnotesize{}{}-611}\tabularnewline
\emph{\footnotesize{}{}Average daily bridge traffic}{\footnotesize{} } & {\footnotesize{}{}-601}  & {\footnotesize{}{}-440}  & {\footnotesize{}{}-545}\tabularnewline
\hline 
\end{tabular}{\footnotesize{}{} }%
\begin{tabular}{cccc}
\hline 
{\footnotesize{}{}Attribute }\emph{\footnotesize{}{}} & \multicolumn{3}{c}{{\footnotesize{}{}Loss }\emph{\footnotesize{}{}(\$ million)}}\tabularnewline
\hline 
 & {\footnotesize{}{}Scenario I}  & {\footnotesize{}{}Scenario II}  & {\footnotesize{}{}Scenario III}\tabularnewline
\hline 
\emph{\footnotesize{}{}Property losses}{\footnotesize{} } & {\footnotesize{}{}-953}  & {\footnotesize{}{}-699}  & {\footnotesize{}{}-754}\tabularnewline
\emph{\footnotesize{}{}Fatalities}{\footnotesize{} } & {\footnotesize{}{}-658}  & {\footnotesize{}{}-60}  & {\footnotesize{}{}-215}\tabularnewline
\emph{\footnotesize{}{}Air Departures}{\footnotesize{} } & {\footnotesize{}{}-264}  & {\footnotesize{}{}-19}  & {\footnotesize{}{}-117}\tabularnewline
\emph{\footnotesize{}{}Average daily bridge traffic}{\footnotesize{} } & {\footnotesize{}{}-205}  & {\footnotesize{}{}-714}  & {\footnotesize{}{}-86}\tabularnewline
\hline 
\end{tabular}{\footnotesize{}{}}\\
 
\par\end{center}

\begin{center}
{\footnotesize{}{}}%
\begin{tabular}{cccc}
\hline 
{\footnotesize{}{}Attribute }\emph{\footnotesize{}{}} & \multicolumn{3}{c}{{\footnotesize{}{}Loss }\emph{\footnotesize{}{}(\$ million)}}\tabularnewline
\hline 
 & {\footnotesize{}{}Scenario I}  & {\footnotesize{}{}Scenario II}  & {\footnotesize{}{}Scenario III}\tabularnewline
\hline 
\emph{\footnotesize{}{}Property losses}{\footnotesize{} } & {\footnotesize{}{}-455}  & {\footnotesize{}{}-199}  & {\footnotesize{}{}-442}\tabularnewline
\emph{\footnotesize{}{}Fatalities}{\footnotesize{} } & {\footnotesize{}{}-314}  & {\footnotesize{}{}-103}  & {\footnotesize{}{}-401}\tabularnewline
\emph{\footnotesize{}{}Air Departures}{\footnotesize{} } & {\footnotesize{}{}-106}  & {\footnotesize{}{}-402}  & {\footnotesize{}{}-851}\tabularnewline
\emph{\footnotesize{}{}Average daily bridge traffic}{\footnotesize{} } & {\footnotesize{}{}-946}  & {\footnotesize{}{}-116}  & {\footnotesize{}{}-100}\tabularnewline
\hline 
\end{tabular}{\footnotesize{}{} }%
\begin{tabular}{cccc}
\hline 
{\footnotesize{}{}Attribute }\emph{\footnotesize{}{}}  & \multicolumn{3}{c}{{\footnotesize{}{}Loss }\emph{\footnotesize{}{}(\$ million)}}\tabularnewline
\hline 
 & {\footnotesize{}{}Scenario I}  & {\footnotesize{}{}Scenario II}  & {\footnotesize{}{}Scenario III}\tabularnewline
\hline 
\emph{\footnotesize{}{}Property losses}{\footnotesize{} } & {\footnotesize{}{}-697}  & {\footnotesize{}{}-56}  & {\footnotesize{}{}-550}\tabularnewline
\emph{\footnotesize{}{}Fatalities}{\footnotesize{} } & {\footnotesize{}{}-954}  & {\footnotesize{}{}-451}  & {\footnotesize{}{}-795}\tabularnewline
\emph{\footnotesize{}{}Air Departures}{\footnotesize{} } & {\footnotesize{}{}-805}  & {\footnotesize{}{}-271}  & {\footnotesize{}{}-100}\tabularnewline
\emph{\footnotesize{}{}Average daily bridge traffic}{\footnotesize{} } & {\footnotesize{}{}-280}  & {\footnotesize{}{}-423}  & {\footnotesize{}{}-237}\tabularnewline
\hline 
\end{tabular}{\footnotesize{}{} }\\
 
\par\end{center}

\begin{center}
{\footnotesize{}{}}%
\begin{tabular}{cccc}
\hline 
{\footnotesize{}{}Attribute }\emph{\footnotesize{}{}} & \multicolumn{3}{c}{{\footnotesize{}{}Loss }\emph{\footnotesize{}{}(\$ million)}}\tabularnewline
\hline 
 & {\footnotesize{}{}Scenario I}  & {\footnotesize{}{}Scenario II}  & {\footnotesize{}{}Scenario III}\tabularnewline
\hline 
\emph{\footnotesize{}{}Property losses}{\footnotesize{} } & {\footnotesize{}{}0}  & {\footnotesize{}{}-334}  & {\footnotesize{}{}-100}\tabularnewline
\emph{\footnotesize{}{}Fatalities}{\footnotesize{} } & {\footnotesize{}{}-356}  & {\footnotesize{}{}-352}  & {\footnotesize{}{}-133}\tabularnewline
\emph{\footnotesize{}{}Air Departures}{\footnotesize{} } & {\footnotesize{}{}-293}  & {\footnotesize{}{}-469}  & {\footnotesize{}{}-475}\tabularnewline
\emph{\footnotesize{}{}Average daily bridge traffic}{\footnotesize{} } & {\footnotesize{}{}-268}  & {\footnotesize{}{}-351}  & {\footnotesize{}{}-464}\tabularnewline
\hline 
\end{tabular}{\footnotesize{}{} }%
\begin{tabular}{cccc}
\hline 
{\footnotesize{}{}Attribute }\emph{\footnotesize{}{}} & \multicolumn{3}{c}{{\footnotesize{}{}Loss }\emph{\footnotesize{}{}(\$ million)}}\tabularnewline
\hline 
 & {\footnotesize{}{}Scenario I}  & {\footnotesize{}{}Scenario II}  & {\footnotesize{}{}Scenario III}\tabularnewline
\hline 
\emph{\footnotesize{}{}Property losses}{\footnotesize{} } & {\footnotesize{}{}-244}  & {\footnotesize{}{}-451}  & {\footnotesize{}{}-190}\tabularnewline
\emph{\footnotesize{}{}Fatalities}{\footnotesize{} } & {\footnotesize{}{}-459}  & {\footnotesize{}{}-478}  & {\footnotesize{}{}-320}\tabularnewline
\emph{\footnotesize{}{}Air Departures}{\footnotesize{} } & {\footnotesize{}{}-140}  & {\footnotesize{}{}-221}  & {\footnotesize{}{}-121}\tabularnewline
\emph{\footnotesize{}{}Average daily bridge traffic}{\footnotesize{} } & {\footnotesize{}{}-1}  & {\footnotesize{}{}-113}  & {\footnotesize{}{}-293}\tabularnewline
\hline 
\end{tabular}{\footnotesize{}{} } 
\par\end{center}

\end{document}